\documentclass[a4paper,11pt]{article}

\usepackage{times}                       
\usepackage{CJK,CJKnumb,CJKulem}         
\usepackage{color}
\usepackage{appendix}
\usepackage{diagbox}
\usepackage{multirow}
\usepackage{booktabs}
\usepackage{threeparttable}             
\usepackage{caption}
\usepackage{chngcntr}
\usepackage{hyperref}
\hypersetup{hypertex=ture,
	colorlinks=true,
	linkcolor=blue,
	anchorcolor=blue,
	citecolor=blue}

\usepackage{mathrsfs}
\usepackage{qcircuit}

\counterwithin{figure}{section}
\counterwithin{table}{section}

\newcommand\comment[1]{}

\newcommand{\bra}[1]{\ensuremath{\left\langle#1\right|}}
\newcommand{\ket}[1]{\ensuremath{\left|#1\right\rangle}}

\usepackage{amsmath,amsthm,amsfonts,amssymb,bm} 
\usepackage{graphicx,psfrag}                    
\usepackage{subfigure}
\numberwithin{equation}{section}
\usepackage[capitalise]{cleveref}
\usepackage{breqn}
\usepackage{authblk}  
\usepackage{hyperref} 
\newtheorem{theorem}{Theorem}[section] 
\newtheorem{lemma}{Lemma}[section]     
\newtheorem{remark}{Remark}[section]
\newtheorem{example}{Example}[section]

\newcommand\ddelta\bigtriangledown
\newcommand\ld\lambda
\newcommand\Ld\Lambda

\def\cal{\mathcal}

\def\to{{\longrightarrow}}

\makeatletter
\def\thanks#1{\protected@xdef\@thanks{\@thanks
		\protect\footnotetext{#1}}}
\makeatother
\begin{document} 
	\date{}
	
	\title{Quantum Circuits for the Black-Scholes equations via
		Schr\"{o}dingerisation  }
	
	\author[1,2,3,4]{Shi Jin \thanks{ shijin-m@sjtu.edu.cn}}
	\author[2]{Zihao Tang \thanks{ momo\textunderscore 0@sjtu.edu.cn}}
	\author[5]{Xu Yin \thanks{ yinx@sxu.edu.cn}}
	\author[1,2,3]{Lei Zhang \thanks{ 
			lzhang2012@sjtu.edu.cn}}
	\affil[1]{School of Mathematical Sciences, Shanghai Jiao Tong University, Shanghai 200240, China}
	\affil[2]{Institute of Natural Sciences, Shanghai Jiao Tong University, Shanghai 200240, China}
	\affil[3]{Ministry of Education, Key Laboratory in Scientific and Engineering Computing, Shanghai Jiao Tong University, Shanghai 200240, China.}
	\affil[4]{Shanghai Jiao Tong University
		Shanghai-Chongqing Institute of Artificial Intelligence, Chongqing 401329, China}
	\affil[5]{School of Mathematics and Statistics, Shanxi University, Taiyuan 030006, China}
	
	\renewcommand*{\Affilfont}{\small\it} 
	\renewcommand\Authands{ and } 

	\maketitle
	
	\begin{abstract}
		In this paper, we construct  quantum circuits for the Black-Scholes equations, a cornerstone of financial modeling, based on a quantum algorithm that overcome the cure of high dimensionality. Our approach leverages the Schrödingerisation technique, which converts linear partial and ordinary differential equations with non-unitary dynamics into a system evolved by unitary dynamics. This is achieved through a warped phase transformation that lifts the problem into a higher-dimensional space, enabling the simulation of the Black-Scholes equation on a quantum computer. We will conduct a thorough complexity analysis to highlight the quantum advantages of our approach compared to existing algorithms. The effectiveness of our quantum circuit is substantiated through extensive numerical experiments.
		
	\end{abstract}
	
	\section{Introduction}
	The Black-Scholes equation, first introduced in 1973, is an important mathematical model in the field of derivatives pricing, which lays a theoretical foundation for derivatives pricing and promotes further research on mathematical modeling of financial markets. Numerical methods for solving the Black-Scholes equations include finite difference methods, finite element methods, Monte Carlo methods and Fourier spectral methods \cite{ace2019,an2021,boy1997,chen2012,roul2020,val2014}. While these methods offer practical advantages and address specific computational requirements, significant challenges remain in tackling high-dimensional multi-asset derivative pricing problems due to the curse of dimensionality. 
	
	In this context, quantum computing has emerged as a transformative solution for financial mathematics, potentially enabling polynomial or exponential speedups in complex derivative pricing problems like Black-Scholes while improving computational efficiency and accuracy. Rebentrost et al. \cite{reb2018} presented quantum algorithms for Monte Carlo estimation of the Black-Scholes equations, achieving a square-order speedup over classical sampling. In addition, Stamatopoulos et al. \cite{stam2020} extended these capabilities using gate-based methods to price diverse options, highlighting quantum computing's versatility for complex financial instruments. Calderer et al. \cite{cald2021} further advanced the field with a quantum circuit based on the one-dollar Black-Scholes model, proving both feasibility and practical advantages for option pricing. These developments underscore quantum computing's significant potential to revolutionize option pricing methodologies.
	
	This paper is based on  a quantum simulation algorithm for Black-Scholes equations using the Schr\"{o}dingerisation method. The technique transforms financial PDEs with non-unitary dynamics into Schr\"{o}dinger-type systems through a warped phase transformation, enabling efficient quantum solutions to derivative pricing problems. The method's versatility is demonstrated through implementations in both qubit-based \cite{jin2022,jin2022_1} and continuous-variable \cite{jin2023} quantum computing frameworks.

	Despite significant progress in quantum computing for PDEs, the explicit design of quantum circuits for solving partial differential equations remains a key challenge. Recent breakthroughs include the work of Sato et al. \cite{sato2024}, who developed a scalable Bell-basis approach for quantum circuit construction tailored to wave and heat equations. Building on Schr\"{o}dingerisation techniques, Hu et al. \cite{hu2024} further demonstrated the method's versatility by designing quantum circuits for both heat and advection equations.
	
	In this work, we present quantum circuits for the Black-Scholes equation based on the Schr\"{o}dingerisation method. Our contributions include:
	\begin{itemize}
		\item Rigorous derivation of error bounds and computational complexity estimates;
		\item Demonstration of polynomial speedup over classical methods in high-dimensional settings;
		\item Numerical verification on Qiskit, confirming both convergence properties and theoretical complexity analysis.
	\end{itemize}
	
	The paper is organized as follows: Section 2 introduces fundamental notations. Sections 3 and 4 present quantum circuit implementations for one-dimensional and $d$-dimensional Black-Scholes equations, respectively. Section 5 provides a complexity analysis, while Section 6 discusses numerical results. Conclusions are drawn in Section 7.
	
	\textbf{Notation:} We use $O$ for upper bounds and $\tilde{O}$ to ignore logarithmic terms. Pauli matrices are denoted by $X$, $Y$, and $Z$.
	
	\section{Notations for finite difference operators }
	
	Consider the one-dimensional domain $\Omega:=(0,L)$ where $L$ is  length of the domain. The domain $\Omega$ is uniformly subdivided into   $N_x=2^{n_x}$ intervals with the interval $h=L/N_x$.
	The discrete function ${\bf u}:=[u_0, u_1, \cdots, u_{Nx-1}]$  can express as a quantum state 
	$|u\rangle:=\sum_{j=0}^{2^{n_x}-1} u_j|j\rangle$.
	
	For the qubit-based system, the finite difference operators can be represented as matrix product operators (MPOs) using the following three 2 × 2 matrices
	\begin{equation*}
		\sigma_{01}:=|0\rangle \langle 1|=
		\begin{bmatrix}
			0& 1 \\
			0& 0
		\end{bmatrix},
		\quad
		\sigma_{10}:=|1\rangle \langle 0|=
		\begin{bmatrix}
			0& 0 \\
			1& 0
		\end{bmatrix},
		\quad 
		I:=|0\rangle \langle 0|+|1\rangle \langle 1|.
	\end{equation*}
	The shift operators can be defined as
	\begin{equation}\label{shift}
		\begin{aligned}
			S^{-}:&=\sum_{j=1}^{2^{n_x}-1} |j-1\rangle \langle j|
			=\sum_{j=1}^{{n_x}}  I^{\otimes(n_x-j)} \otimes \sigma_{01}\otimes \sigma_{10}^{\otimes(j-1)} \triangleq \sum_{j=1}^{{n_x}} s_j^-,\\
			S^{+}:=(S^-)^\dagger&=\sum_{j=1}^{2^{n_x}-1} |j\rangle \langle j-1|
			=\sum_{j=1}^{{n_x}}  I^{\otimes(n_x-j)} \otimes \sigma_{10}\otimes \sigma_{01}^{\otimes(j-1)} \triangleq \sum_{j=1}^{{n_x}} s_j^+.
		\end{aligned}
	\end{equation}
	
	The forward and backward finite difference operators are,
	\begin{equation}\label{shift_op}
		(D^+{\bf u})_j=\frac{u_{j+1}-u_j}{h}, \quad 
		(D^+{\bf u})_j=\frac{u_{j}-u_{j-1}}{h}, \quad j=0,1,\cdots,N_x-1.
	\end{equation}
	The difference operator with Dirichlet boundary conditions can be expressed as follows,
	\begin{equation*}
		D^+_D=\frac{S^--I^{\otimes n_x}}{h}, \quad D^-_D=\frac{I^{\otimes n_x}-S^+}{h},\quad
		D^\pm_D=\frac{S^--S^+}{2h},\quad D^\triangle_D=\frac{S^-+S^+-2I^{\otimes n_x}}{h^2}.
	\end{equation*}

	\section{The one-dimensional Black-Scholes equation}
	
	In this section, we implement quantum circuits for the one-dimensional Black-Scholes model of European call options using the Schr\"{o}dingerisation quantum simulation framework.
	
	Consider a scalar field $w$ governed by the one-dimensional Black-Scholes equation,

	\begin{equation}\label{bs}
		\frac{\partial w}{\partial t}+rs \frac{\partial w}{\partial s}+	\frac{\sigma^2 s^2}{2}\frac{\partial^2 w}{\partial s^2}=rw.
	\end{equation}
	This equation is a variable coefficient parabolic convection-diffusion equation.
	Here, $s\in[s_{\min}, s_{\max}]$  represents the stock price and  $w$ is the option price. $r$ and $\sigma$ are the risk-free interest rate and volatility, respectively, which are typically positive constants.

	Through  variable substitution $ s=e^{x}, -\infty<x<\infty$ and $\tau=T-t$, equation \eqref{bs} is converted  into a positive constant-coefficient partial differential equation
	\begin{equation}\label{bs_model}
		\left\{
		\begin{aligned}
			&		\frac{\partial w}{\partial \tau }=\left(r-\frac{\sigma^2}{2}\right)\frac{\partial w}{\partial x}+
			\frac{\sigma^2}{2}\frac{\partial^2 w}{\partial x^2}-rw,\\
			&	w(x,0)=w_0,\\
		\end{aligned}
		\right.
	\end{equation}
	where $w_0=w(0,x)=\max(e^x-K,0)$.
	The European option contract is mainly considered, and the exercise price is assumed to be $K$.
	The Dirichlet boundary conditions are represented as
	\[w(x,\tau)  = 0,\quad x\rightarrow -\infty, \quad w(x,\tau)  = e^x -K e^{-r\tau},\quad x\rightarrow +\infty.\]
	
	Alternatively, we may impose the following mixed boundary conditions,
	\begin{equation*}
		w(x,\tau) = 0 \quad \text{as} \quad x \to -\infty, \quad \frac{\partial w(x,\tau)}{\partial x} = e^x \quad \text{as} \quad x \to \infty,
	\end{equation*}
	which remain consistent with the Dirichlet boundary condition formulation of model \eqref{bs_model}.
	
	To facilitate the application of the Schr\"odingerisation method, we introduce the variable substitution $u = w - e^{x}$,
	which transforms both the equation and the Dirichlet boundary conditions to:
	\begin{equation}\label{bs_equation}
		\left\{
		\begin{aligned}
			&	\frac{\partial u}{\partial \tau }=\left(r-\frac{\sigma^2}{2}\right)\frac{\partial u}{\partial x}+\frac{\sigma^2}{2}\frac{\partial^2 u}{\partial x^2}-ru,\\
			&	u(x,0)=w_0-e^x,\\
			&	u(x,\tau)  =  -e^x,\quad x\rightarrow -\infty,\\
			&	u(x,\tau)  =  -K e^{-r\tau},\quad x\rightarrow +\infty.
		\end{aligned}
		\right.
	\end{equation}
	
	Assuming that $x_{\min}$ is a sufficiently small negative number and $x_{\max}$ is a sufficiently large positive number, let the interval $[L, R]$ be the truncation of the logarithmic space of asset prices which can be   uniformly subdivided into   $N_x=2^{n_x}+1$ intervals with the interval $h=(R -L)/N_x$.
	The equation \eqref{bs_equation} can be solved using the central  difference method in classical simulation as follows:
	\begin{equation}\label{finite}
		\frac{d{\bf u}}{d\tau }=\left(r-\frac{\sigma^2}{2}\right) D_D^{\pm} {\bf u} +\frac{\sigma^2}{2} D_D^{\Delta}{\bf u}-rI {\bf u} +{\bf b}\triangleq A{\bf u}+{\bf b}, 
		\quad {\bf u}(0)={\bf u}_0,
	\end{equation}
	where ${\bf u}=[u(x_1,\tau),u(x_2,\tau),\cdots,u(x_{N_x-1},\tau)]^T$, and $u(x_0,\tau)=-e^{x_0}$, $u(x_{N_x},\tau)=-K e^{-r\tau}$. 
	The initial value is given by
	\[{\bf u}_0=[u(x_1,0),u(x_2,0),\cdots,u(x_{N_x-1},0)]^T\]
	and    
	\begin{equation}\label{value:b}
		{\bf b}=\left[\left(-\frac{r-\frac{\sigma^2}{2}}{2h}+\frac{\sigma^2}{2h^2}\right)u(x_0,\tau),0,\cdots,0, \left(\frac{r-\frac{\sigma^2}{2}}{2h}+\frac{\sigma^2}{2h^2}\right) u(x_{N_x},\tau)\right]^T.
	\end{equation}
	
	\begin{remark}
		For mixed boundary conditions,	the internal nodes are treated in the same way as the  Dirichlet condition, only virtual point $x_{N_x+1}=x_{N_x}+h$ need to be introduced at the right boundary of the interval. The first derivative is discretized by central difference method
		\begin{equation*}
			\frac{u_{N_x+1}-u_{N_x-1}}{2h}=0.
		\end{equation*}
		It is reasonable to assume  the above grid discretization at $x_{N_x}$, then
		\begin{equation*}
			\begin{aligned}
				\frac{d u_{N_x}}{dt}&=\left (-\frac{\sigma^2}{h^2}-r\right)u_{N_x} +\left(\frac{r-\frac{\sigma^2}{2}}{2h}+\frac{\sigma^2}{2h^2}\right)u_{N_x+1}+ \left(-\frac{r-\frac{\sigma^2}{2}}{2h}+\frac{\sigma^2}{2h^2}\right)u_{N_x-1}.
			\end{aligned}
		\end{equation*}
		Then with \eqref{finite}, the ODE system with mixed boundary conditions is obtained.
		Its matrix dimension is one order higher than the discrete matrix $A$ of the Dirichlet boundary condition model.
		
	\end{remark}
	
	\begin{remark}
		In Equation \eqref{value:b}, since $x_{\text{0}}$ is a sufficiently small negative number, we can take the boundary condition as $u(x_{0},\tau) = 0$, i.e., ${\bf b}=\left[0,0,\cdots,0, \left(\frac{r-\frac{\sigma^2}{2}}{2h}+\frac{\sigma^2}{2h^2}\right) u(x_{N_x},\tau)\right]^T$.
	\end{remark}

	\subsection{Schr\"{o}dingerisation}
	Introduce the auxiliary variable ${\bf R}(t)$ and let ${\bf R}(0)=(1,\cdots,1)^{T}, B = \text{diag}(0,\cdots, 0, -K(\frac{1}{2}\frac{\sigma^{2}}{h^{2}} + \frac{r-\frac{\sigma^{2}}{2}}{2h}))$. By dilating the system ${\bf u} \rightarrow \bar {\bf u}={\bf u}\otimes |0\rangle +{ \bf R}(t)\otimes|1\rangle$, one can derive
	\begin{equation}\label{dilate}
		\frac{d \bar {\bf u}}{d\tau}= C\bar {\bf u},\quad 
		C=\begin{bmatrix}
			A& B \\
			{\bf 0}& -rI
		\end{bmatrix}.
	\end{equation}
		In this case, $ C\neq  C^{\dagger}$, we decompose the matrix as
		\begin{equation}\label{decompose}
			\begin{aligned}
				A&=A_1+iA_2,\quad   C=C_1+iC_2,\\
				C_1&=( C+ C^\dagger)/{2}=C_1^\dagger, \quad 	C_2=( C- C^\dagger)/{2i}=C_2^\dagger,\\
				C_1&=\begin{bmatrix}
					A_1&  \frac{1}{2}B\\
					\frac{1}{2}B& rI
				\end{bmatrix}=|0\rangle \langle 0| \otimes A_1+ |1\rangle \langle 1|\otimes rI+   X\otimes\frac{1}{2} B,\\
				C_2&=\begin{bmatrix}
					A_2&   \frac{1}{2i}B\\
					- \frac{1}{2i}B& {\bf 0}
				\end{bmatrix}=|0\rangle \langle 0|\otimes A_2+ Y\otimes \frac{1}{2}B.
			\end{aligned}
		\end{equation}
		
		Applying the Schr\"odingerisation method with the warped phase transformation $\mathbf{v}(\tau, p) = e^{-p} \bar{\mathbf{u}}(\tau)$, $p > 0$, we obtain
		\begin{equation}\label{warp}
			\partial_\tau {\bf v}(\tau,p)=C_1 {\bf v}(\tau,p)+i C_2 {\bf v}(\tau,p)=-C_1 \partial_p {\bf v}(\tau,p)+i C_2 {\bf v}(\tau,p), \quad p>0,
		\end{equation}
		and this equation is extended to $p<0$ with initial data ${\bf v}(0,p)=e^{-|p|}\bar {\bf u}(0)$.
		Fourier transform with respect to $p$ leads to
		\begin{equation*}
			\partial_\tau  \hat{{\bf v}}(\tau,\eta)=i(\eta C_1+C_2) \hat {{\bf v}}(\tau,\eta)\triangleq i H_{BS}\hat {{\bf v}}(\tau ,\eta),
		\end{equation*}
		with $H_{BS}$ being Hermitian. 
		
		\begin{remark}\label{remark}
			For cash or worthless options with Dirichlet or mixed boundary conditions where the option price is equal to 0 at the boundaries, satisfying ${\bf b} = 0$, 
			the Schr\"{o}dingerisation method can be directly applied to equation  \eqref{finite} without dilation. By  directly converting $A$ into a Hermitian matrix, the Schr\"{o}dingerisation method can be used.
		\end{remark}
		
		\subsection{Discretization}
		
		Discretize the domain in $p$ as $-\pi L_p=p_0<\cdots<p_{N_p} =\pi L_p$ with a mesh size of $\Delta p=2\pi L_p/N_p$. Define $\eta_k=\left(k-\frac{N_p}{2}\right)/L_p$ for $k=0,1,\cdots, N_p-1$.
		The variables ${\bf v}$ and $\hat{\bf v}$ are discretized as ${\bf v}(t):=[v(x_j,p_k,t)]_{j,k}$ and $\hat {\bf v}(t):=[\hat v(x_j,\eta_k,t)]_{j,k}$ with $1\leq j\leq N_x-1$, $0\leq k\leq N_p-1$.
		The initial conditions can be represented as
		\begin{equation}\label{initial}
			\begin{aligned}
				{\bf v}(0):&= \bar{\bf u}(0) \otimes {\bf p}=[\bar u(x_1,0),\bar u(x_2,0),\cdots,\bar u(x_{2N_x-2},0)] \otimes [e^{-|p_0|},e^{-|p_1|},\cdots,e^{-|p_{Np-1}|}],\\
				\hat {\bf v}(0):&=\bar {\bf u}(0) \otimes \eta=[\bar u(x_1,0),\bar u(x_2,0),\cdots,\bar u(x_{2N_x-2},0)] \otimes
				\left[ \frac{1}{\eta_0^2+1}, \frac{1}{\eta_1^2+1},\cdots, \frac{1}{\eta_{N_p-1}^2+1}\right].
			\end{aligned}
		\end{equation}
		The Hamiltonian $H_{BS}$ is discretized as
		\begin{equation}\label{hbs}
			\begin{aligned}
				H_{BS}&=C_1\otimes D_\eta+C_2\otimes I^{\otimes n_p}\\
				&=\left(|0\rangle \langle 0| \otimes\left(\frac{\sigma^{2}}{2}D^{\Delta}_{D} - rI^{\otimes n_{x}}\right)
				-|1\rangle \langle 1|\otimes rI^{\otimes n_{x}} + X\otimes\frac{1}{2}B \right)\otimes D_\eta \\
				&+ \left( |0\rangle \langle 0|\otimes \frac{1}{i} \left(r-\frac{\sigma^2}{2}\right) D_D^{\pm}+Y\otimes\frac{1}{2}B\right)\otimes I^{\otimes n_p}\\
				&=\sum_{k=0}^{N_p-1}\left(k-\frac{N_p}{2}\right)\left(|0\rangle \langle 0|\otimes(\frac{\sigma^2}{2}{ H}_1 - \frac{r}{L_{p}}I^{\otimes n_{x}}) - |1\rangle \langle 1|\otimes\frac{r}{L_{p}}I^{\otimes n_{x}}\right)\otimes|k\rangle \langle k|\\
				&	+\sum_{k=0}^{N_p-1}\left(k-\frac{N_p}{2}\right)/L_p \left(	X\otimes \frac{1}{2}B\right)\otimes|k\rangle \langle k|\\
				&+ \left( |0\rangle \langle 0|\otimes \left(r-\frac{\sigma^2}{2}\right) { H}_2  +Y\otimes \frac{1}{2}B \right) \otimes I^{\otimes n_p},
			\end{aligned}
		\end{equation}
		where
		\begin{equation}\label{h123}
			\begin{aligned}
				{ H}_1&=\gamma_1\left[\sum_{j=1}^{n_x}(s_j^-+s_j^+)-2I^{\otimes n_x}\right],\quad \gamma_1=\frac{1}{h^2L_p},\\
				{ H}_2&=-i\gamma_2\left[\sum_{j=1}^{n_x}(s_j^- -s_j^+)\right],\quad \gamma_2=\frac{1}{2h}.\\
			\end{aligned}
		\end{equation}
		
		\subsection{{Quantum circuit}}

		By employing the first-order Lie-Trotter-Suzuki decomposition, we can approximate the time evolution operator $U_{BS}(\tau):=\exp(i \tau{ H}_{BS} )$,
		
		\begin{dmath}\label{ubs}
			U_{BS}(\tau) \approx \exp\left(i\tau \left( |0\rangle \langle 0|\otimes \left(r-\frac{\sigma^2}{2}\right) { H}_2  +Y\otimes \frac{1}{2}B  \right) \otimes I^{\otimes n_p}\right) \\
			\exp\left(i\tau \sum_{k=0}^{N_p-1}\left(k-\frac{N_p}{2}\right)\left(|0\rangle \langle 0|\otimes(\frac{\sigma^2}{2}{ H}_1- \frac{r}{L_{p}}I^{\otimes n_{x}}) - |1\rangle \langle 1|\otimes\frac{r}{L_{p}}I^{\otimes n_{x}} \right)\otimes|k\rangle \langle k| \right)\\
			\exp\left(i\tau \sum_{k=0}^{N_p-1}\left(k-\frac{N_p}{2}\right)/L_p\left(X\otimes \frac{1}{2}B\right)\otimes|k\rangle \langle k| \right)\\
			= \left(\ket{0}\bra{0}\otimes U_2\left( \left(r-\frac{\sigma^2}{2}\right) \tau \right) + \ket{1}\bra{1}\otimes I^{\otimes n_{x}}\right) U_2^{(1)}(\tau)
			\otimes I^{\otimes n_p}\\
			\cdot
			\sum_{k=0}^{N_p-1} \left(\left(\ket{0}\bra{0}\otimes \text{Ph}(-\tau\frac{r}{L_{p}})U_1\left(\frac{\sigma^2}{2}\tau\right) + \ket{1}\bra{1}\otimes \text{Ph}(-\tau\frac{r}{L_{p}})I^{\otimes n_{x}}\right) {U_1^{(1)}(\tau)}\right)^{k-N_p/2}\otimes|k\rangle \langle k|\\
			\triangleq  U_{*}(\tau),
		\end{dmath}
		
		where 
		\begin{equation}\label{eq:1}
			\begin{aligned}
				U_j(\tau)&:= \exp\left(i\tau H_j\right) ,\quad j=1,2,   \\
				U_1^{(1)}(\tau)&:=\exp\left(i\tau/L_p \left(X\otimes \frac{1}{2}B\right) \right),\\
				U_2^{(1)}(\tau)&:=\exp\left(i\tau Y\otimes \frac{1}{2}B\right),\\
				{\rm Ph}(\theta)&:=e^{i\theta} I^{\otimes n_x},
			\end{aligned}
		\end{equation}
		and we have used the following properties:
		\begin{equation*}
			\exp\left(\sum_k A_k \otimes |k\rangle \langle k|\right)=\sum_k \exp(A_k) \otimes |k\rangle \langle k|.
		\end{equation*}

		Noticing that $B$ is an $n\times n$ diagonal matrix with only one nonzero element at the  $(n,n)$ position, and  $U_1^{(1)}$ and $U_2^{(1)}$ can be simulated by multi-controlled $R_{x}, \ R_{y}$ gates.	Again using the first-order Lie-Trotter-Suzuki decompositon, $U_1(\tau)$ is approximated as
		
		\begin{equation}\label{hamiltion1}
			\begin{aligned}
				U_1(\tau)&=\exp\left(i\tau \gamma_1\left[\sum_{j=1}^{n_x}(s_j^-+s_j^+)-2I^{\otimes n_x}\right]\right)
				\\
				&= \exp(-2i\gamma_1 \tau) \exp\left(i\tau \gamma_1 \left[\sum_{j=1}^{n_x}\left(s_j^-+s_j^+\right)\right]\right) \\
				&\triangleq  {\rm Ph}(-2\gamma_1 \tau) U_1^{(2)}(\tau).
			\end{aligned}
		\end{equation}

		$U_2$ is approximated as
		
		\begin{equation}\label{hamiltion3}
			\begin{aligned}
				U_2(\tau)&=\exp\left(i\tau \gamma_2 \left[\sum_{j=1}^{n_x}-i(s_j^--s_j^+)\right]\right)
				\triangleq  U_2^{(2)}(\tau).
			\end{aligned}
		\end{equation}
		According to \cite{hu2024}, one can utilize an explicit quantum circuit which solves the general form $\gamma \sum_{j=1}^{n_x} (e^{i\lambda} s_j^- +e^{-i\lambda} s_j^+)$ with $\lambda, \gamma \in {\mathbb{R}}$,
		\begin{equation*}
			e^{i\lambda} s_j^- +e^{-i\lambda} s_j^+
			=I^{\otimes (n_x-j)} \otimes \left(e^{i\lambda}|0\rangle| 1\rangle^{\otimes{j-1}} \langle1| \langle 0|
			^{\otimes(j-1)} +e^{-i\lambda}|1\rangle| 0\rangle^{\otimes{j-1}} \langle 0| \langle 1|
			^{\otimes(j-1)}\right).
		\end{equation*}
		Define the unitary matrix Bell basis $B_j (\lambda)$ such that
		\begin{equation*}
			\begin{aligned}
				&	B_j(\lambda)| 0\rangle | 1\rangle^{\otimes(n-1)}=\frac{|0\rangle | 1\rangle
					^{\otimes(j-1)}+e^{-i\lambda}| 1\rangle | 0\rangle
					^{\otimes(j-1)}}{\sqrt{2}},\\
				&	B_j(\lambda)| 1\rangle | 1\rangle^{\otimes(n-1)}=\frac{|0\rangle | 1\rangle
					^{\otimes(j-1)}+e^{-i\lambda}| 1\rangle | 0\rangle
					^{\otimes(j-1)}}{\sqrt{2}}.
			\end{aligned}
		\end{equation*}
		Furthermore, $B_j(\lambda)$ is constructed by
		\begin{equation*}
			B_j(\lambda):=\left( \displaystyle\prod_{m=1}^{j-1} {\rm CNOT}_m^j \right) P_j(-\lambda) H_j,
		\end{equation*}
		where $H_j$ is the  Hadamard gate acting on the
		$j$-th qubit. $P_j(\lambda)$ is the Phase gate acting on the
		$j$-th qubit as
		\begin{equation*}
			P_j(\lambda)=\begin{bmatrix}
				1& 0 \\
				{ 0}& e^{i\lambda}
			\end{bmatrix}.
		\end{equation*}
		The time evolution operator is formulated as
		\begin{equation*}
			\begin{aligned}
				\exp(i\gamma \tau 	(e^{i\lambda} s_j^- +e^{-i\lambda} s_j^+) )&=
				I^{\otimes(n_x-j)} \otimes B_j(\lambda) {\rm CRZ}_j^{1,\cdots,j-1} (-2\gamma \tau) B_j(\lambda)^\dagger\\
				&\triangleq 	I^{\otimes(n_x-j)} \otimes W_{j} (\gamma \tau,\lambda),
			\end{aligned}
		\end{equation*}
		where ${\rm CRZ}_{j}^{1,\cdots, j-1}(\theta):=\exp(-i\theta Z_j/2)\otimes |1\rangle \langle 1|^{\otimes(j-1)}+ I\otimes(I^{\otimes(j-1)}-|1\rangle \langle 1|^{\otimes(j-1)})$
		is the multicontroled
		$\rm RZ$ gate acting on the $j$-th qubit controlled by $1, \cdots, j-1$-th qubits. 
		
		Applying the first-order Lie-Trotter-Suzuki decomposition, we have 
		\begin{equation*}
			\exp\left(i\gamma \tau \sum_{j=1}^{n_x}	(e^{i\lambda} s_j^- +e^{-i\lambda} s_j^+) \right)\approx
			\displaystyle\prod_{j=1}^{n_x} \exp(i\gamma \tau 	(e^{i\lambda} s_j^- +e^{-i\lambda} s_j^+) )
			=\displaystyle\prod_{j=1}^{n_x} I^{\otimes(n_x-j)} \otimes W_{j} (\gamma \tau,\lambda).
		\end{equation*}
		
		From this point, we set $\lambda=0$, $\gamma=\gamma_1$ and $\lambda=-\frac{\pi}{2}$,  $\gamma=\gamma_2$, respectively, to simulate $U_1^{(2)}$ and $U_2^{(2)}$ in \eqref{hamiltion1} and \eqref{hamiltion3}. Then
		$U_1$ and $U_2$ can be constructed explicitly as
		\begin{equation}\label{appro_V}
			\begin{aligned}
				U_1(\tau)&\approx{\rm Ph}(-2\gamma_1 \tau)  \displaystyle\prod_{j=1}^{n_x} I^{\otimes(n_x-j)} \otimes W_j(\gamma_1\tau,0)  
				\triangleq V_1(\tau),\\
				U_2(\tau)&\approx \displaystyle\prod_{j=1}^{n_x} I^{\otimes(n_x-j)} \otimes W_j(\gamma_2\tau,-\frac{\pi}{2})\triangleq V_2(\tau).
			\end{aligned}
		\end{equation}

		Ultimately, the time evolution operator ${U_{BS}(\tau )}$ is approximated by
		\begin{equation}\label{vbs}
			\begin{aligned}
				V_{BS}(\tau)&:= \tilde{V}_2(\tau)
				\otimes I^{\otimes n_p} 
				\sum_{k=0}^{N_p-1} \tilde{V}_1(\tau)^{k-N_p/2}
				\otimes|k\rangle \langle k|,
			\end{aligned}
		\end{equation}
		where
		\begin{align*}
			&\tilde{V}_1(\tau)= \left(\ket{0}\bra{0}\otimes \text{Ph}(-\tau\frac{r}{L_{p}})V_1\left(\frac{\sigma^2}{2}\tau\right) + \ket{1}\bra{1}\otimes \text{Ph}(-\tau\frac{r}{L_{p}})I^{\otimes n_{x}}\right) {U_1^{(1)}(\tau)},\\
			&\tilde{V}_2(\tau)=\left(\ket{0}\bra{0}\otimes V_2\left( \left(r-\frac{\sigma^2}{2}\right) \tau \right) + \ket{1}\bra{1}\otimes I^{\otimes n_{x}}\right) U_2^{(1)}(\tau).
		\end{align*}

		In order to minimize the gate count during the implementation, we utilize the binary representation
		of integers $k=(k_{n_p-1 \cdots k_0.})=\sum_{m=0}^{n_p-1} k_m 2^m$ to obtain
		
		\begin{equation*}
			\begin{aligned}
				{{\tilde{V}}_{BS}(\tau)}
				=
				\tilde{V}_2(\tau)
				\tilde{V}_1(\tau)^{-N_p/2}
				\otimes I^{\otimes n_p}
				{{\prod}'}_{m=0}^{n_p-1} \left(\tilde{V}_1(\tau)^{2^m}
				\otimes|1\rangle \langle 1|
				+I^{\otimes n_x}\otimes|0\rangle \langle 0|\right),
			\end{aligned}
		\end{equation*}
		where the primed product ${{\prod}'}$
		denotes the regular matrices product for the first register (consisting
		of $n_x+1$ qubits) and the tensor product for the second register (consisting of $n_p$ qubits).
		Assuming that ${{\tilde{V}}_{BS}(\tau)}$can be implemented with a cost independent of $k$, this approach is exponentially more efficient than directly implementing ${{{V}}_{BS}(\tau)}$.
		The detailed quantum circuits for the approximating unitary matrices  are depicted in Figure \ref{W_source}-\ref{VBS}.

		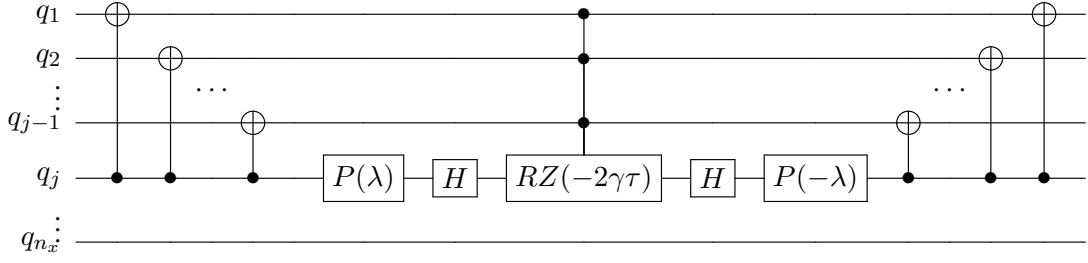
\begin{figure}[h!] 
			\centerline{
				\Qcircuit @C=1em @R=.7em
				{   \lstick{q_1}& \targ & \qw & \qw & \qw &\qw&\qw&\qw & \ctrl{4} &\qw&\qw &\qw &\qw &\qw &\targ &\qw \\
					\lstick{q_2}& \qw & \targ & \qw & \qw &\qw&\qw&\qw & \ctrl{3} &\qw &\qw &\qw &\qw &\targ &\qw&\qw \\
					\lstick{\vdots}&  &   & \cdots &  && &  & & & & &\cdots&\\
					\lstick{q_{j-1}}& \qw & \qw&\qw  & \targ & \qw &\qw&\qw & \ctrl{1} &\qw &\qw &\targ &\qw &\qw&\qw&\qw\\
					\lstick{q_j}& \ctrl{-4} & \ctrl{-3} &\qw& \ctrl{-1} & \qw & \gate{P(\lambda)} &\gate{H}  &\gate{RZ(-2\gamma \tau)} &\gate{H} & \gate{P(-\lambda)}  & \ctrl{-1}&\qw & \ctrl{-3} & \ctrl{-4} &\qw \\
					\lstick{\vdots}\\
					\lstick{q_{n_x}}& \qw&\qw & \qw & \qw & \qw &\qw&\qw &\qw &\qw &\qw &\qw &\qw &\qw&\qw &\qw}
			}
			\caption{Quantum circuit for $W_{j}(\gamma \tau,\lambda)$}
			\label{W_source}
		\end{figure}

		\begin{figure}[h!] 
			\centerline{
				\Qcircuit @C=1em @R=.7em
				{    \lstick{q_1}&\gate{X} &\ctrl{1}     &\qw  &\ctrl{1}  &\ctrl{1}  &\ctrl{1}&\gate{X} &\ctrl{1}&\gate{R_{x}}&\qw\\
					\lstick{q_2}& \qw &\gate{ W_1(\gamma_1 \tau,0)} & \cdots& \multigate{2}{ W_{n_x-1}(\gamma_1 \tau,0)} & \multigate{3}{ W_{n_x}(\gamma_1 \tau,0)}  &\gate{{\rm Ph}(\gamma_1 \tau)}&\qw&\gate{{\rm Ph}(\tau r/L_p )}&\ctrl{0} &\qw\\
					\lstick{\vdots}&& &\cdots&  \ghost{ W_{n_x-1}(\gamma_1 \tau,0)} &\ghost{ W_{n_x}(\gamma_1 \tau,0)}&&&&&   & \\
					\lstick{q_{n_x}}& \qw& \qw& \qw &  \ghost{ W_{n_x-1}(\gamma_1 \tau,0)}  &\ghost{ W_{n_x}(\gamma_1 \tau,0)} & \qw &\qw&\qw&\ctrl{0}&\qw\\
					\lstick{q_{n_x+1}}& \qw& \qw& \qw &  \qw  &\ghost{ W_{n_x}(\gamma_1 \tau,0)} & \qw 		&\qw&	\qw&\ctrl{-4}&\qw			
					\gategroup{2}{3}{5}{6}{.7em}{--}  
				}
			}
			\caption{Quantum circuit for 	$\tilde{V}_1(\tau)$}
			\label{V1_source}
		\end{figure}

		\begin{figure}[h!] 
			\centerline{
				\Qcircuit @C=1em @R=.7em
				{ \lstick{q_1}&\gate{X}  &\ctrl{1}  & \qw &\ctrl{1}   &\ctrl{1} &\gate{X}&\qw &\gate{R_{y}}&\qw &\qw\\
					\lstick{q_2}& \qw &\gate{W_1(\gamma_2 \tau,-\frac{\pi}{2})} & \qw& \multigate{2}{ W_{n_x-1}(\gamma_2 \tau,-\frac{\pi}{2})} & \multigate{3}{ W_{n_x}(\gamma_2 \tau,-\frac{\pi}{2})} &\qw  &\qw&\ctrl{0}&\qw&\qw\\
					\lstick{\vdots}&& & \cdots & \ghost{W_{n_x-1}(\gamma_2 \tau,-\frac{\pi}{2})} &  \ghost{W_{n_x}(\gamma_2 \tau,-\frac{\pi}{2})} & && && \\
					\lstick{q_{n_x}}& \qw& \qw& \qw &  \ghost{ W_{n_x-1}(\gamma_2 \tau,-\frac{\pi}{2})} &  \ghost{W_{n_x}(\gamma_2 \tau,-\frac{\pi}{2})} & \qw &\qw &\ctrl{0}&\qw&\qw\\
					\lstick{q_{n_x+1}}& \qw& \qw &  \qw&\qw &  \ghost{W_{n_x}(\gamma_2 \tau,-\frac{\pi}{2})} & \qw & \qw&\ctrl{-4}  & \qw&\qw 
					\gategroup{2}{3}{5}{6}{.7em}{--} 
				}
			}
			\caption{Quantum circuit for 	$\tilde{V}_2(\tau)$}
			\label{V3_source}
		\end{figure}

		\begin{figure}[h!] 
			\centerline{
				\Qcircuit @C=.5em @R=.7em
				{\lstick{q} &\qw&\qw&{/^{n_x +1} }\qw&\qw &\qw &\qw &  \gate{\tilde{V}_1(\tau)} & \gate{\tilde{V}_1^2(\tau)} &\qw& \cdots & & \gate{\tilde{V}_1^{2^{n_p-1}}(\tau)}&\gate{\tilde{V}_1^{-2^{n_p-1}}(\tau)} &   \gate{\tilde{V}_2(\tau)} &\qw\\
					\lstick{	k_0}&\qw&\qw&\qw &\qw&\qw &\qw &  \ctrl{-1}&\qw &\qw& \cdots &   & \qw & \qw  &\qw&\qw \\
					\lstick{k_1}&\qw&\qw&\qw &\qw &\qw&\qw & \qw& \ctrl{-2}&\qw& \cdots & &\qw& \qw&\qw& \qw\\
					\lstick{	\vdots}\\
					\lstick{k_{n_p-1}}&\qw&\qw&\qw &\qw &\qw&\qw &\qw&\qw &\qw& \cdots &\qw&\ctrl{-4} & \qw  & \qw &\qw  
				}
			}
			\caption{Quantum circuit for ${\tilde{V} }_{BS}$}
			\label{VBS}
		\end{figure}
		
		For time $T=N \tau$, one can obtain $|\hat {\bf v}(T)\rangle:= V_{BS}^N |\hat {\bf v}(0)\rangle $, which means
		the operator $V_{BS}(\tau)$ is applied $N$ times.
		We can perform an inverse quantum Fourier transform on it to retrieve $|{\bf v}(T)\rangle$. Then a measurement
		(the projection $P := I^{\otimes n_x+1}\otimes|p_k^{\diamond}\rangle\langle p_k^{\diamond}|$) is taken to select only the $|p_k^{\diamond}\rangle$ part of	the state $|{\bf u}(T)\rangle$. 
		The full circuit to implement the Schr\"{o}dingerisation method is shown in Figure \ref{schro}.
		
		\begin{remark}
			
			According to 
			Theorem 5.1 in \cite{ma2024},  $ p_k^{\diamond} > \max(\lambda_n(C_1)T, 0)$, where $\lambda_n(C_1) $ is  the max eigenvalues of matrix $C_1$,  we can recovery of the original variables ${\bf u}$. While  $- rI$ is a negative number multiplied by the identity matrix  and $B$ in \eqref{dilate} 
			is a matrix with only one non-zero entry, we just focus on the eigenvalues of $A_1$.
			
		\end{remark}

		\begin{figure}[h!] 
			\centerline{
				\Qcircuit @C=1em @R=2.4em {
					\lstick{} &  & &&&& &&	\lstick{ {\bf u}\otimes \ket{0}+{\bf b}\otimes \ket{1}} &\qw
					&{/^{n_x +1} }&\qw&\qw& \multigate{1}{V_{BS}^N}&\qw&\qw& {{\bf v}(T)}\\
					\lstick{} && && \lstick{\bf p}&\qw&\qw&\qw&\qw&\qw{/^{n_p } }&\qw&\qw&\gate{{\cal {QFT}}} &\ghost{V_{BS}}&\gate{\cal{IQFT}}& \qw&& {P>0}
					\inputgroupv{1}{2}{.8em}{.8em}{{{\bf v}(0)}}
				}
			}
			\caption{Quantum circuit for the Schr\"{o}dingerisation method, where the measurement requires
				only projection onto $P>0$ and $\cal{QFT (IQFT )}$ denotes the (inverse) quantum Fourier
				transform. }
			\label{schro}
		\end{figure}
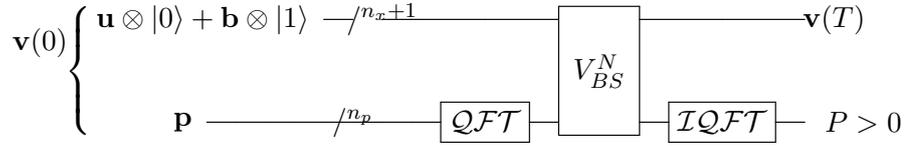
		
		\section{The d-dimension Black-Scholes equation}
		
		Denote $s_1, s_2, \cdots , s_d$ as the prices of $d$ risky assets. $W$ is the option price based on the underlying assets $s_1, s_2, \cdots , s_d$, and its pricing law satisfies the $d$-dimensional Black-Scholes equation:
		\begin{equation*}\label{bs_d}
			\frac{\partial U}{\partial t}+\frac{1}{2} \sum_{m,n=1}^{d} \sigma_m \sigma_n \rho_{mn} s_m s_n \frac{\partial^2 U}{\partial{s_m}\partial {s_n}}+\sum_{m=1}^d r s_m\frac{\partial U}{\partial s_{m}} -rU=0,
		\end{equation*}
		where $t\in [0,T]$, and ${\bf s}=[s_1,\cdots,s_d]^T\in (0,s_{\max}^1)\times\cdots \times (0,s_{\max}^d)$, $\rho_{m,n}\in[-1,1], (m,n=1,\cdots,d)$.

		Using the similar technique applied for the 1d case, set $s_i=e^{x_i}, \tau=T-t$, one can derive:
		\begin{equation}\label{bsd}
			\frac{\partial U}{\partial t}=\frac{1}{2} \sum_{m,n=1}^{d} \sigma_m \sigma_n \rho_{mn} \frac{\partial^2 U}{\partial{x_m}\partial {x_n}}+\sum_{m=1}^d \left(r -\frac{\sigma_m^2}{2}\right)\frac{\partial U}{\partial x_{m}} -rU.
		\end{equation}
		For high dimensional Black-Scholes equation,
		mixed boundary conditions are often used in application. 
		Assume the following boundary condition for call option
		\begin{equation*}
			U(x_1, \cdots, x_{j-1},L, x_{j+1},\cdots, x_d)=0,\quad 
			\frac{\partial U}{\partial x_j}(x_1, \cdots, x_{j-1},R, x_{j+1},\cdots, x_d)=0.
		\end{equation*}							 
		Then	the equation \eqref{bsd} can be solved using the central  difference method as follows:
		\begin{equation}\label{BS-d}
			\frac{d{ U}}{d\tau }=\sum_{m=1}^{d}\left(r-\frac{\sigma^2_{m}}{2}\right) (D_D^{\pm})_{m} U +\sum_{m,n=1}^d \frac{\sigma_m \sigma_n \rho_{mn}}{2} (D_D^{\Delta})_{mn} U-r U
			\triangleq {\bf A}{U}
			,
		\end{equation}
		where ${U}=[U_{j_1,\cdots,j_d}]_{1\leq j_i\leq N_{x}-1}, (D_D^{\pm})_{m}=I^{\otimes (m-1)n_x }\otimes 
		D_D^{\pm} \otimes I^{\otimes(d-m-1)n_x}.$

		\begin{remark}\label{remark1}
			In \eqref{BS-d},
			the operator $ D_D^{\Delta}$ is expressed by  the central  difference formula
			\begin{equation*}
				(D_D^{\Delta})_{mn}u=	\left\{
				\begin{aligned}
					&	\frac{1}{h^2}(u_{m+1,n}-2u_{m,n}+u_{m-1,n}), \quad m=n,\\
					&	\frac{1}{4h^2}(u_{m+1,n+1}-u_{m+1,n-1}-u_{m-1,n+1}-u_{m-1,n-1}), \quad m\neq n.
				\end{aligned}
				\right.
			\end{equation*}
			For the case $m=n$, one can obtain $(D_D^{\Delta})_{mn}=I^{\otimes (m-1)n_x}\otimes D_D^{\Delta}\otimes I^{\otimes (d-m-1)n_x}$, where $D_D^{\Delta}$ is defined in \eqref{shift_op}. 
			while 	for $m\neq n$, $(D_D^{\Delta})_{mn}=I^{\otimes (m-1)n_x}\otimes \frac{1}{2h}S^+ \otimes I^{\otimes (d-m-n)n_x} \otimes D_D^{\pm} \otimes I^{\otimes (d-n-1)n_x}+
			I^{\otimes (m-1)n_x}\otimes \frac{1}{2h}S^- \otimes I^{\otimes (d-m-n)n_x} \otimes D_D^{\pm} \otimes I^{\otimes (d-n-1)n_x}$, where $S^+$ and $S^-$ are shift operators defined in \eqref{shift}.
		\end{remark}
		
		\subsection{Schr\"{o}dingerisation}

		Most banks use a simplified model for ordinary basket options that ignores weak correlations.
		Thus, in this section, we consider the  simple case $\rho_{mn}=0, m\neq n$, then the homogeneous ordinary differential equation system \eqref{BS-d} is 
		
		\begin{equation}\label{eq:d-dim}
			\frac{d{  U}}{d\tau }={\bf A}{  U},\quad { U }(0)={   U}_0.
		\end{equation}
		
		In this case, $ {\bf A}\neq  -{\bf A}^{\dagger}$, we decompose the matrix as
		\begin{equation*}
			\begin{aligned}
				{\bf A}&={\bf A}_1+i{\bf A}_2,\\
				{\bf A}_1&=( {\bf A}+ {\bf A}^\dagger)/{2}={\bf A}_1^\dagger, \quad 	{\bf A}_2=( {\bf A}- {\bf A}^\dagger)/{2i}={\bf A}_2^\dagger,\\
			\end{aligned}
		\end{equation*}
		where ${\bf A}_i=\sum_{j=1}^d  I^{\otimes (j-1) n_x}\otimes A_i\otimes I^{\otimes (d-j) n_x}$, ($i=1,2$), $A_i$ is defined in \eqref{decompose}.
		
		Applying the Schr\"{o}dingerisation method, the warped phase transformation ${ \cal V}(\tau, p)=e^{-p}  \bar{  U}(\tau)$, $p>0$, satisfies,
		\begin{equation*}
			\partial_\tau { \cal V}(\tau,p)={\bf A}_1 {  \cal V}(\tau,p)+i {\bf A}_2 {  \cal V}(\tau,p)=-{\bf A}_1 \partial_p {  \cal V}(\tau,p)+i {\bf A}_2 { \cal  V}(\tau,p), \quad p>0,
		\end{equation*}
		and this equation is extended to $p<0$ with initial data ${   \cal V}(0,p)=e^{-|p|}{ U}(0)$.
		After the Fourier transform with respect to $p$, we obtain,
		\begin{equation*}
			\partial_\tau  \hat{ \cal  V}(\tau,\eta)=i(\eta {\bf A}_1+{\bf A}_2) \hat { \cal V}(\tau,\eta)\triangleq i {\bf H}_{BS}\hat {  \cal V}(\tau ,\eta),
		\end{equation*}
		with ${\bf H}_{BS}$ being Hermitian.

		\subsection{Quantum circuit}
		In the $d$-dimensional case, one has
		\begin{equation}\label{dhbs}
			\begin{aligned}
				{\bf H}_{BS}&={\bf A}_1\otimes D_\eta+{\bf A}_2\otimes I^{\otimes n_p}\\
				&=\sum_{k=0}^{N_p-1}\left(k-\frac{N_p}{2}\right) 
				\sum_{m=1}^d \left( \frac{\sigma_m^2  \rho_{mm}}{2}({  H}_1)_{m}\right) \otimes|k\rangle \langle k|\\
				&+  \sum_{m=1}^{d}\left(r-\frac{\sigma^2_{m}}{2}\right)({ H}_2)_m    \otimes I^{\otimes n_p},
			\end{aligned}
		\end{equation}
		where 
		$({ H}_1)_{mm}:=I^{\otimes(m-1)n_x} \otimes { H}_1 \otimes I^{\otimes(d-m)n_x}$,
		$({ H}_2)_m:=I^{\otimes(m-1)n_x} \otimes { H}_2 \otimes I^{\otimes(d-m)n_x}, H_i (i=1,2)$ defined in \eqref{h123}.
		
		By employing the first-order Lie-Trotter-Suzuki decomposition, we can approximate the time evolution operator ${\bf U}_{BS}(\tau):=\exp(i {\bf H}_{BS}\tau )$,
		
		\begin{equation}\label{dubs}
			\begin{aligned}
				{\bf U}_{BS}(\tau) \approx& \exp\left(i\tau \sum_{m=1}^{d}\left(r-\frac{\sigma^2_{m}}{2}\right)({ H}_2)_m  \otimes I^{\otimes n_p}\right)\\					&	\exp\left(i\tau\sum_{k=0}^{N_p-1}\left(k-\frac{N_p}{2}\right) 
				\sum_{m=1}^d \frac{\sigma_m^2 \rho_{mm}}{2}({  H}_1)_{m}\right) \otimes|k\rangle \langle k|\\
				=&\displaystyle\prod_{m=1}^{d} {U}_2\left( \left(r-\frac{\sigma^2_{m}}{2}\right)\tau  \right)_m 
				\otimes I^{\otimes n_p}\\
				&\sum_{k=0}^{N_p-1} \left(\displaystyle\prod_{m=1}^{d}  { U}_1\left( \frac{\sigma_m^2 \rho_{mm}}{2}\tau\right)_{m}  \right)^{k-N_p/2}
				\otimes|k\rangle \langle k|\\
				\triangleq&  {\bf U}_{*}(\tau).
			\end{aligned}
		\end{equation}

		Similar to the  1-dimensional case, using  the first-order Lie-Trotter-Suzuki decomposition, the time evolution operator ${{\bf U}_{BS}(\tau )}$ is approximated by
		\begin{equation}\label{dvbs}
			\begin{aligned}
				{\bf V}_{BS}(\tau)&:=   \tilde{\bf V}_2 (\tau)\otimes I^{\otimes n_p}
				\sum_{k=0}^{N_p-1}\tilde{\bf V}_1(\tau)^{k-N_p/2}
				\otimes|k\rangle \langle k|,
			\end{aligned}
		\end{equation}
		where 	 \[
		\tilde{\bf V}_1(\tau)=\displaystyle\prod_{m=1}^{d} { V}_1\left(\frac{\sigma_m^2 \rho_{mm}}{2}\tau\right)_{m},\]
		
		\[\tilde{\bf V}_2(\tau)=\displaystyle\prod_{m=1}^{d} { V}_2 \left(\left(r-\frac{\sigma_m^2}{2}\right)\tau\right)_m,\]
		and $ V_i(\tau)_m=I^{\otimes(m-1)n_x} \otimes V_i\otimes I^{\otimes (d-m)n_x}$ with $V_i, (i=1,2)$ defined in  \eqref{appro_V}.
		
		Furthermore,we utilize the binary representation
		of integers $k=(k_{n_p-1 \cdots k_0.})=\sum_{m=0}^{n_p-1} k_m 2^m$ to minimize the gate count
		\begin{equation*}
			\begin{aligned}
				{{\bf {V}}_{BS}(\tau)}	=(\tilde{\bf V}_2(\tau)
				\tilde{\bf V}_1(\tau)^{-N_p/2}
				\otimes I^{\otimes n_p})
				{{\prod}'}_{m=0}^{n_p-1} \left(\tilde{\bf V}_1(\tau)^{2^m}
				\otimes|1\rangle \langle 1|
				+I^{\otimes d n_x}\otimes|0\rangle \langle 0|\right).
			\end{aligned}
		\end{equation*}

		For	the $d$-dimension Black-Scholes equation,	the detailed quantum circuits for the approximating unitary matrices  are depicted in Figure \ref{Vd11}-\ref{VBSd}.
		The full circuit to implement the Schr\"{o}dingerisation method is shown in Figure \ref{schro} with $dn_x+n_p$ qubits.
		
		\begin{figure}[h!] 
			\centerline{
				\Qcircuit @C=1em @R=.7em
				{\lstick{q^1}  & \qw&{/^{n_x} }\qw& \qw & \multigate{5}{\tilde {\bf V}_1(\tau)} &\qw& && \gate{V_1(\frac{\sigma_1^2\rho_{11}}{2} \tau)} &\qw&\qw&\qw&\qw\\
					\lstick{q^2}& \qw&{/^{n_x} }\qw & \qw &\ghost{\tilde {\bf V}_1\tau)}&\qw&&&\qw & \gate{V_1(\frac{\sigma_2^2\rho_{22}}{2} \tau)}&\qw&\qw&\qw\\
					\lstick{	\vdots}\\
					\lstick{q^j}& \qw&{/^{n_x} }\qw& \qw &\ghost{\tilde {\bf V}_1(\tau)}&\qw&{:=} &&\qw&\qw& \gate{V_1(\frac{\sigma_j^2\rho_{jj}}{2} \tau)}&\qw&\qw\\ 
					\lstick{	\vdots}\\
					\lstick{q^d}& \qw&{/^{n_x} }\qw& \qw&\ghost{\tilde {\bf V}_1(\tau)}&\qw&&&\qw&\qw&\qw & \gate{V_1(\frac{\sigma_d^2\rho_{dd}}{2} \tau)}&\qw
				}
			}
			\caption{Quantum circuit for $\tilde {\bf V}_1(\tau)$}
			\label{Vd11}
		\end{figure}
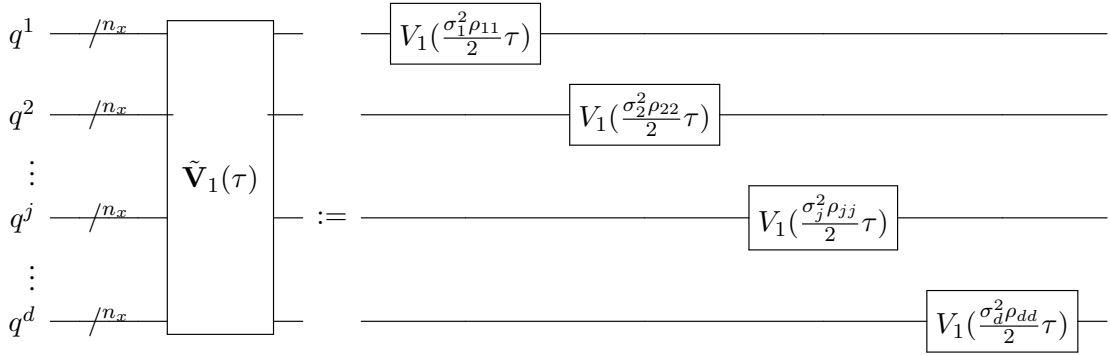

		\begin{figure}[h!] 
			\centerline{
				\Qcircuit @C=1em @R=.7em
				{\lstick{q^1} & \qw &{/^{n_x} }\qw& \qw & \multigate{5}{\tilde {\bf V}_2(\tau)} &\qw && &\qw& \gate{V_2((r-\frac{\sigma_1^2}{2}) \tau)}  & \qw&\qw& \qw& \qw&\qw\\
					\lstick{q^2}& \qw&{/^{n_x} }\qw & \qw &\ghost{\tilde {\bf V}_2(\tau)}&\qw && &\qw&\qw & \gate{V_2((r-\frac{\sigma_2^2}{2}) \tau)}  & \qw&\qw& \qw& \qw\\
					\lstick{	\vdots}\\
					\lstick{q^j}& \qw&{/^{n_x} }\qw& \qw &\ghost{\tilde {\bf V}_2(\tau)}&\qw&{:=}&&\qw&\qw&\qw & \gate{V_2((r-\frac{\sigma_j^2}{2}) \tau)}  & \qw&\qw& \qw\\ 
					\lstick{	\vdots}\\
					\lstick{q^d}& \qw&{/^{n_x} }\qw& \qw&\ghost{\tilde {\bf V}_2(\tau)}& \qw&&& \qw&\qw&\qw&\qw&\gate{V_2((r-\frac{\sigma_d^2}{2}) \tau)}  & \qw&\qw
				}
			}
			\caption{Quantum circuit for $\tilde {\bf V}_2(\tau)$}
			\label{Vd3}
		\end{figure}
		
		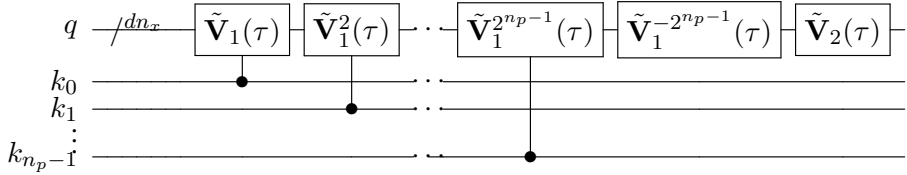
\begin{figure}[h!] 
			\centerline{
				\Qcircuit @C=.5em @R=.7em
				{\lstick{q} &\qw&\qw&{/^{dn_x } }\qw&\qw &\qw &\qw &  \gate{\tilde{\bf V}_1(\tau)} & \gate{\tilde{\bf V}_1^2(\tau)} &\qw& \cdots & & \gate{\tilde{\bf V}_1^{2^{n_p-1}}(\tau)}&\gate{\tilde{\bf V}_1^{-2^{n_p-1}}(\tau)} &   \gate{\tilde{\bf V}_2(\tau)} &\qw\\
					\lstick{	k_0}&\qw&\qw&\qw &\qw&\qw &\qw &  \ctrl{-1}&\qw &\qw& \cdots &   & \qw & \qw  &\qw&\qw \\
					\lstick{k_1}&\qw&\qw&\qw &\qw &\qw&\qw & \qw& \ctrl{-2}&\qw& \cdots & &\qw& \qw&\qw& \qw\\
					\lstick{	\vdots}\\
					\lstick{k_{n_p-1}}&\qw&\qw&\qw &\qw &\qw&\qw &\qw&\qw &\qw& \cdots &\qw&\ctrl{-4} & \qw  & \qw &\qw  
				}
			}
			\caption{Quantum circuit for ${\tilde{\bf V} }_{BS}$}
			\label{VBSd}
		\end{figure}
		
		\begin{remark}
			For the case  $\rho_{mn}\neq 0, m\neq n$, we notice that the Schr\"{o}dingerisation method is still applicable.
			While  using 
			the Lie-Trotter-Suzuki decomposition,
			the decomposition of the coupled quantum gates $\exp(S^+\otimes D_D^{\pm})$ should be considered 
			as discussed in Remark \ref{remark1}.
			
		\end{remark}			
		
		\section{Complexity analysis}
		
		In this subsection, we estimate the complexity of the quantum circuits of the one-dimensional and $d$-dimensional Black-Scholes equation constructed in the previous
		sections and demonstrate their scalability.
		
	\begin{lemma}\label{lem}
		For the Schrödinger equation
		$
		\frac{d|{\bf u}(t)\rangle}{dt} = iH_{BS}|{\bf u}(t)\rangle,
		$
		with the Hamiltonian \( H_{BS} \) as specified in Equation \eqref{hbs}, the time evolution operator \( U_{BS}(\tau) = \exp(iH_{BS}\tau) \) over a time step \( \tau \) can be effectively approximated by the unitary operator \( V_{BS}(\tau) \), which is detailed in Equation \eqref{vbs}. The quantum circuit realization of \( V_{BS}(\tau) \) is depicted in Figures \ref{W_source} - \ref{VBS}. In addition, the approximation error evaluated in the operator norm, is bounded  by
		\begin{equation}\label{err}
			\|U_{BS}(\tau) - V_{BS}(\tau)\| \leq \frac{1}{4} \tau^2 n_x (N_p \gamma_1^2 + 2 \gamma_2^2 + 2N_p \gamma_1\gamma_2)c^2,
		\end{equation}
		where \( N_p = 2^{n_p} \) and \( N_x - 1 = 2^{n_x} \) correspond to the number of grid points for the variables \( p \) and \( x \), respectively. The parameters \( \gamma_1 \) and \( \gamma_2 \) are defined in Equation \eqref{h123}, and \( c \) is given by \( c = \max\left(\frac{\sigma^2}{2}, r - \frac{\sigma^2}{2}\right) \).
	\end{lemma}

		\begin{proof}
			Recalling the definition of $U_1(\tau), U_2(\tau)$, in equation \eqref{hamiltion1} and \eqref{hamiltion3}, $V_1(\tau), V_2(\tau)$ in equation \eqref{appro_V}.
			According to the theory of the Trotter splitting error with commutator scaling \cite{child} and
			the analysis in \cite{sato2024}, we have
			
			\begin{equation*}
				\begin{aligned}
					&	\|U_1(\tau)-V_1(\tau)\|\\
					=&\left\|\exp\left(i\tau \gamma_1 \sum_{j=1}^{n_x}\left(s_j^-+s_j^+\right)\right)-
					\displaystyle\prod_{j=1}^{n_x} I^{\otimes(n_x-j)} \otimes W_j(\gamma_1\tau,0)
					\right\|\\
					\leq &\frac{\gamma_1^2\tau^2}{2} \sum_{j=1}^{n_x} \sum_{j'=j+1}^{n_x} 
					\|[(s_j^-+s_j^+), (s_{j'}^-+s_{j'}^+)]\|
					\leq \frac{\gamma_1^2\tau^2(n_x-1)}{2},
				\end{aligned}
			\end{equation*}
			where for the last inequality, we have applied the commutator results (see e.g.\cite{hu2024} )
			\[\sum_{j=1}^{n_x} \sum_{j'=j+1}^{n_x} 
			\|[(s_j^-+s_j^+), (s_{j'}^-+s_{j'}^+)]\|=n_x-1.\]			
			Similarly, we have
			\begin{equation*}
				\|U_2(\tau)-V_2(\tau)\|	\leq \frac{\gamma_2^2\tau^2(n_x-1)}{2}.
			\end{equation*}
			
			The Trotter error of the first decomposition equation \eqref{ubs} is
			\begin{equation*}
				\begin{aligned}
					&\|U_{BS}(\tau)-U_*(\tau)\|\\
					\leq& \frac{1}{2}\left\|\left[\sum_{k=0}^{N_p-1}\left(k-\frac{N_p}{2}\right)|0\rangle \langle0|\otimes {H}_1\otimes|k\rangle \langle k|, |0\rangle \langle0|\otimes {H}_2 \otimes I^{\otimes n_p}\right]\right\|\\
					\leq & \frac{1}{2} N_p \gamma_1  \gamma_2  n_x \left(r-\frac{\sigma^2}{2}\right) \frac{\sigma^2}{2}.
				\end{aligned}									
			\end{equation*}
			
			
			The error between $U_*$ and $V_{BS}$ is
			\begin{equation*}
				\begin{aligned}
					&\|U_*(\tau)-V_{BS}(\tau)\|\\
					\leq& \left\| {U}_2\left(\left(r-\frac{\sigma}{2}^2\right)\tau\right)-\tilde{V}_2(\tau) \right\| 
					+\displaystyle\max_{0\leq k\leq N_p-1}\left| k-\frac{N_p}{2}\right| \left\|   {U}_1\left(\frac{\sigma^2}{2}\right)-\tilde{V}_1(\tau)\right\|\\
					\leq & \frac{\gamma_2^2 \tau^2 n_x}{2} \left(r-\frac{\sigma^2}{2}\right)^2 
					+\frac{N_p}{2} \frac{\gamma_1^2 \tau^2 n_x}{2} \left(\frac{\sigma^2}{2}\right)^2.
				\end{aligned}
			\end{equation*}
			To sum up, we have
			\begin{equation*}
				\begin{aligned}
					\|U_{BS}(\tau)-V_{BS}(\tau)\|\leq &\|U_{BS}(\tau)-U_*(\tau)\|+\|U_*(\tau)-V_{BS}(\tau)\|\\
					\leq & \frac{1}{4} \tau^2 n_x (N_p \gamma_1^2+2 \gamma_2^2+2N_p \gamma_1\gamma_2)c^2,
				\end{aligned}
			\end{equation*}
			where $c=\max\left(\frac{\sigma^2}{2}, r-\frac{\sigma^2}{2}\right).$
		\end{proof}
		
		\begin{lemma}\label{gate}
			The approximated time evolution operator $V_{BS}(\tau )$ in  \eqref{vbs} can be implemented using ${\cal Q}_{single}=O(N_pn_x)$ single-qubit gates and at most ${\cal Q}_{V_{BS}} = O(N_pn_x^2)$ {\rm CNOT} gates for
			$n_x \geq 3$.
		\end{lemma}
		
		\begin{proof}
			There are
			a $\tilde{V}_2(\tau)$ gate, a maximum of $2^{n_p-1} \tilde{V}_1$ gates and $\sum_{m=0}^{n_p-1} 2^m=2^{n_p}-1$ controlled $\tilde{V}_1$ gates
			in $V_{BS}(\tau )$.
			Furthermore, $\tilde{V}_1$ consists of
			2 phase gate ${\rm Ph}$,  $W_j$, $(j = 1, \cdots, n_x)$, and a          multi-controlled $R_x$ gate.
			$\tilde{V}_2$ consists of $W_j$, $(j = 1, \cdots, n_x)$, and a 
			multi-controlled $R_y$ gate.
			$ W_j$ can be decomposed into a multi-controlled ${\rm RZ}$ gate
			(a controlled RZ gate for $j=2$), 2 Hadmard gates and $2(j-1)$ CNOT gates in the case $\lambda=0$.
			while $\lambda\neq 0$, decomposition of $W_j$ have  two more $P$ gates.
			
			Noticing that we need two more $X$ gates to transfer the control state from $|0\rangle$ to
			$|1\rangle$ and transfer back.
			Hence the number of single-qubit gates is 
			\begin{equation*}
				\begin{aligned}
					{\cal Q}_{single}=2^{n_p-1}  (2n_x+2n_x+2)
					+2n_x+2n_x+2
					=O(N_p n_x).
				\end{aligned}
			\end{equation*}
			As for CNOT gates,	according to \cite{sato2024} \cite{vale}, a multicontrolled
			$RZ$ or $RX$ gate with $(j - 1)$ control qubits can be decomposed into single-qubit gates	and at most $(16j - 40)$ CNOT gates. Thus, we have,
			\begin{equation*}
				{\cal Q}_{V_2}=	{\cal Q}_{V_1}=2+\sum_{j=3}^{n_x+1}(16j-40)+2n_x+2n_x+\sum_{j=3}^{n_x} 2(j-1)*8+16(n_x+1)-40=16 n_x^2-4n_x-14.
			\end{equation*}
			
			The controlled $V_1$ gate consists of a controlled phase gate and controlled  $W_j$, $(j = 1, \cdots, n_x)$. $c-W_j$
			can be decomposed into a multi-controlled RZ gate (a controlled RZ gate for $j = 1$), two
			controlled H gates, and $2(j - 1)$ controlled CNOT gates. Therefore, the CNOT gates
			in  controlled $V_{1}$  gate can be estimated by
			\begin{equation*}
				{\cal Q}_{c-V_1}=2+\sum_{j=3}^{n_x+2} (16j-40)+2n_x+2n_x+\sum_{j=3}^{n_x} 2(j-1)*24=16(n_x+2)-40+2=32n_x^2-12n_x-66.
			\end{equation*}
			Then, the CNOT gates in  $V_{BS}$ can be estimated
			\begin{equation*}
				\begin{aligned}
					{\cal Q}_{BS}={\cal Q}_{V_2} + 2^{n_p-1} {\cal Q}_{V_1}+(2^{n_p}-1){\cal Q_{c-V_1}}
					=O(N_p n_x^2).
				\end{aligned}
			\end{equation*}
		\end{proof}

		\begin{lemma}\label{lem1}
			For time $T$,  the time evolution operator
			$U_{BS}(T) = \exp(iH_{BS}T)$  can be implemented on a $(n_x +1+ n_p)$-qubit system
			using quantum circuits with $O(N_p  T^2 n_x^2 (N_p \gamma_1^2+2 \gamma_2^2+2N_p \gamma_1\gamma_2)c^2/\varepsilon)$ single-qubit gates and  $O(N_p n_x^3  T^2  (N_p \gamma_1^2+2 \gamma_2^2+2N_p \gamma_1\gamma_2)c^2/\varepsilon)$ CNOT gates, within an additive error of
			$\varepsilon$. 
		\end{lemma}
		\begin{proof}
		The time interval $T$ is partitioned into $N$ subintervals, with $N = T/\tau$. It is assumed that there exists a small $\varepsilon$,
			\begin{equation}\label{err1}
				\|U_{BS}^N(\tau)-V_{BS}^N(\tau)\|\leq N	\|U_{BS}(\tau)-V_{BS}(\tau)\|\leq \varepsilon.
			\end{equation}
			With the results \eqref{err}, we have
			\begin{equation*}
				\frac{1}{4N} T^2 n_x (N_p \gamma_1^2+2 \gamma_2^2+2N_p \gamma_1\gamma_2)c^2 \leq \varepsilon,
			\end{equation*}
			which can be rearranged as
			\begin{equation*}
				N\geq \frac{1}{4 \varepsilon} T^2 n_x (N_p \gamma_1^2+2 \gamma_2^2+2N_p \gamma_1\gamma_2)c^2.
			\end{equation*}
			According to Lemma \ref{gate}, $V_{BS}$ consists of $O(N_p n_x)$ single-qubit gates and at	most $O(N_p n_x^2)$ CNOT gates.
			Therefore, $V_{BS}^N$ can be implemented using $O(N_p  T^2 n_x^2 (N_p \gamma_1^2+2 \gamma_2^2+2N_p \gamma_1\gamma_2)c^2/\varepsilon)$ single-qubit gates and at
			most $O(N_p n_x^3  T^2  (N_p \gamma_1^2+2 \gamma_2^2+2N_p \gamma_1\gamma_2)c^2/\varepsilon)$ CNOT gates.
		\end{proof}

	\begin{theorem}\label{th1}
		For the one-dimensional Black-Scholes equation  \eqref{bs_equation}, the state $|{\bf u}(t)\rangle$, where ${\bf u}(t)$ is the solution of
		\eqref{finite} with a mesh size $h$, can be prepared with precision $\varepsilon'$ using the Schr\"{o}dingerisation
		method depicted in Figure \ref{VBS} and Figure \ref{schro}. This preparation can be achieved using at most
		$\tilde{O}\left(\frac{T^2 \|{\bf u}(0)\|^3}{h^4\varepsilon^3 \|{\bf u}(T)\|^3}\right) $
		single-qubit gates and {\rm CNOT} gates.
	\end{theorem}

	\begin{proof}
	
	As illustrated in Figure \ref{schro}, the quantum circuit consists of a quantum Fourier transform (QFT), an inverse quantum Fourier transform (IQFT), a unitary operation \( V_{BS}^N(\tau) \), and a projection \( P \). 
	
	First, the (inverse) quantum Fourier transform can be implemented using \( O(n_p^2) \) controlled rotation (CR) gates, which corresponds to \( O(n_p^2) \) CNOT gates \cite{lin2022, nie2010, vale}.
	
	Second, the complexity analysis for simulating the unitary \( V_{BS}^N(\tau) \), as concluded in Lemma \ref{lem1}, can be applied with \( n_x = O(\log(L_x/h)) \) and \( \gamma_1 = \frac{1}{h^2 L_p}, \gamma_2 = \frac{1}{2h} \). This results in a computational cost of \( O\left(\frac{T^2 N_p^2 \log^3(L_x/h) c^2}{h^4 L_p^2 \varepsilon'}\right) \), where \( \varepsilon' \) denotes the desired precision \( |U_{BS}^N(T) - V_{BS}^N(\tau)| \leq \varepsilon' \).
	
	Third, the number of measurements and the error of the output state can be discussed in the context of the given expressions \( |{\hat{\bf v}(T)}\rangle = U_{BS}^N(T)|{\hat{\bf v}(0)}\rangle \) and \( |{\hat{\bf v}_D(T)}\rangle = V_{BS}^N(T) |{\hat{\bf v}(0)}\rangle \). To retrieve \( |{\hat{\bf v}(T)}\rangle \) and \( |{\hat{\bf v}_D(T)}\rangle \), an inverse quantum Fourier transform can be performed on these states, followed by a projection \( P \) onto \( |p_k^\diamond\rangle \) through measurement. Given the relationship \( v(t,x,p) = e^{-p} u(t,x) \) for \( p > 0 \), one can observe that

	\begin{equation*}
		\frac{\|  P   {\bf v}(T)- e^{-p_k^\diamond} {\bf u}(T)|p_k^\diamond\rangle\|}{\|e^{-p_k^\diamond} {\bf u}(T)\|}
		=O\left(\frac{\pi L_p}{N_p}+e^{-\pi L_p}\right),
	\end{equation*}
	which represents the (relative) discretization error over the $p$ variable. It follows by,
	\begin{equation*}
		\begin{aligned}
			&	\left\| \frac{e^{p_k^\diamond}\|{\bf v}(0)\|}{\|{\bf u}(T)\|} P\ket{{\bf v}_D(T)}-\ket{{\bf u}(T)}\ket{p_k^\diamond}  \right\|\\
			\leq &\left\|  \frac{e^{p_k^\diamond}\|{\bf v}(0)\|}{\|{\bf u}(T)\|} P\ket{{\bf v}_D(T)}-\frac{e^{p_k^\diamond}\|{\bf v}(0)\|}{\|{\bf u}(T)\|} P\ket{{\bf v}(T)}\right\|+ \left\| \frac{e^{p_k^\diamond}\|{\bf v}(0)\|}{\|{\bf u}(T)\|} P\ket{{\bf v}(T)}-\ket{{\bf u}(T)}\ket{p_k^\diamond}  \right\|\\
			= & O\left(\frac{e^{p_k^\diamond}\|{\bf u}(0)\|}{\|{\bf u}(T)\|} \varepsilon' +\frac{\pi L_p}{N_p}+e^{-\pi L_p} \right).
		\end{aligned}
	\end{equation*}

	To ensure the precision of the output state, the following conditions are imposed:
	\( L_p = O(\log(1/\varepsilon)) \),
	\( N_p = (L_p / \varepsilon) = \tilde{O}(1/\varepsilon) \),
	\( p_k^{\diamond} = O(1) \),
	\( \varepsilon' = O(\|{\bf u}(T)\|\varepsilon / \|{\bf u}(0)\|) \).
	These conditions are essential for maintaining the desired accuracy in the quantum simulation process.
	The probability of obtaining the desired state is given by \( O(\|{\bf u}(T)\|^2 / \|{\bf u}(0)\|^2) \). Consequently, the number of measurements required to achieve the desired state with high probability is \( O(\|{\bf u}(0)\|^2 / \|{\bf u}(T)\|^2) \). This ensures that the output state is accurately retrieved with the specified precision.

	Combining all the steps, the Schr\"{o}dingerisation method requires at most
	$
	\tilde{O}\left(\frac{T^2 \|{\bf u}(0)\|^3}{h^4\varepsilon^3 \|{\bf u}(T)\|^3}\right)
	$
	single-qubit gates and CNOT gates.
	\end{proof}
		{	We now extend the above discussion to the operators
			acting on a $d$-dimensional domain $\Omega$  for special cases $\rho_{mn}=0, m\neq n$.}
		\begin{lemma}
			
			For  the \( d \)-dimensional Schrödinger equation
			$
			\frac{d|U(t)\rangle}{dt} = i{\bf H}_{BS}|U(t)\rangle,
			$
			the Hamiltonian \( {\bf H}_{BS} \) is specified by Equation \eqref{dhbs}. The time evolution operator \( {\bf U}_{BS}(\tau) = \exp(i{\bf H}_{BS}\tau) \) with a time increment \( \tau \) is approximated by the unitary operator \( {\bf V}_{BS}(\tau) \) as detailed in Equation \eqref{dubs}. The explicit quantum circuit realization of \( {\bf V}_{BS}(\tau) \) is illustrated in Figures \ref{Vd11}--\ref{VBSd}. Additionally, the upper bound on the approximation error, assessed in the operator norm, is given by

			\begin{equation}\label{err_d}
				\begin{aligned}
					\|{\bf U}_{BS}(\tau)-{\bf V}_{BS}(\tau)\|
					\leq  \frac{1}{4} \tau^2 n_x (N_p \gamma_1^2+2 \gamma_2^2+2N_p \gamma_1\gamma_2) \sum_{m=1}^d c_{m}^2,
				\end{aligned}
			\end{equation}
			where $d$ denotes the spatial dimension, $N_p = 2^{n_p}$ and $N_x-1 = 2^{n_x}$ represent the number of grid
			points for the variables $p$ and $x$, respectively. Additionally, $\gamma_1, \gamma_2$ are defined in \eqref{h123} and $c_{m}=\displaystyle\max_{m}\left(r-\frac{\sigma_m^2}{2}, \frac{\sigma_m^2 \rho_{mm}}{2}\right).$
		\end{lemma}
		
		\begin{proof}
			Similar to	the approximation error in  one-dimension  concluded in Lemma \ref{lem},
			for  $d$-dimension case, we have
			\begin{equation*}
				\begin{aligned}
					\left \| \displaystyle\prod_{m=1}^{d} \left({ U}_2\left( \left(r-\frac{\sigma^2_{m}}{2}\right)\tau  \right)_m -\tilde{\bf V}_2 (\tau)_m\right)\right \|
					\leq& \sum_{m=1}^{d} \left\|\left({ U}_2\left(\left(r-\frac{\sigma^2_{m}}{2}\right)\tau  \right)_m -\tilde{\bf V}_2 (\tau)_m\right)\right\|\\
					\leq& \frac{\gamma_2^2\tau^2 (n_x-1)}{2} \sum_{m=1}^{d} \left(r-\frac{\sigma^2_{m}}{2}\right),
				\end{aligned}
			\end{equation*}
			
			\begin{equation*}
				\begin{aligned}
					\left \| \displaystyle\prod_{m,n=1}^{d} \left({U}_1\left( \frac{\sigma_m^2 \rho_{mm}}{4}\tau\right)_{m} -\tilde{\bf V}_1 (\tau)_{m}\right)\right \|
					\leq& \sum_{m=1}^{d} \left\|\left({ U}_1\left( \frac{\sigma_m^2  \rho_{mm}}{4}\tau\right)_{m} -\tilde{\bf V}_1 (\tau)_{m} \right)\right\|\\
					\leq& \frac{\gamma_1^2\tau^2 (n_x-1)}{2} \sum_{m=1}^{d} \frac{\sigma_m^2 \rho_{mm}}{4}.
				\end{aligned}
			\end{equation*}
			
			The Trotter error of the first decomposition Equation \eqref{dubs} is
			\begin{equation*}
				\begin{aligned}
					&\|{\bf U}_{BS}(\tau)-{\bf U}_*(\tau)\|\\
					\leq& \frac{1}{2}\left\|\left[ \displaystyle \sum_{k=0}^{N_p-1}\left(k-\frac{N_p}{2}\right)\sum_{m=1}^{d}  ({ H}_1)_{m} \otimes|k\rangle \langle k|, \displaystyle\sum_{m=1}^{d} ({ H}_2)_m \otimes I^{\otimes n_p}\right]\right\|\\
					\leq & \frac{1}{2} N_p \gamma_1 \gamma_2 2n_x \sum_{m=1}^{d} \left(r-\frac{\sigma^2_m}{2}\right)\frac{\sigma_m ^2 \rho_{mm}}{2}.
				\end{aligned}									
			\end{equation*}
			
			The error between ${\bf U}_*$ and ${\bf V}_{BS}$ is
			\begin{equation*}
				\begin{aligned}
					&\|{\bf U}_*(\tau)-{\bf V}_{BS}(\tau)\|\\
					\leq& \left\|\displaystyle\prod_{m=1}^d ({ U}_2)_m-\tilde{\bf V}_2 \right\| 
					+\displaystyle\max_{0\leq k\leq N_p-1}\left| k-\frac{N_p}{2}\right| \left\| 
					\displaystyle\prod_{m=1}^d  ({ U}_1)_{m}-\tilde{\bf V}_1\right\|\\
					\leq & \frac{\gamma_2^2\tau^2 2n_x}{2} \sum_{m=1}^{d} \left(r-\frac{\sigma^2_{m}}{2}\right)^2 
					+\frac{N_p}{2} \frac{\gamma_1^2\tau^2 2n_x}{2} \sum_{m=1}^{d} \left(\frac{\sigma_m^2 \rho_{mm}}{2}\right)^2.
				\end{aligned}
			\end{equation*}
			To sum up, we have
			\begin{equation*}
				\begin{aligned}
					\|{\bf U}_{BS}(\tau)-{\bf V}_{BS}(\tau)\|\leq &\|{\bf U}_{BS}(\tau)-{\bf U}_*(\tau)\|+\|{\bf U}_*(\tau)-{\bf V}_{BS}(\tau)\|\\
					\leq & \frac{1}{4} \tau^2 n_x (N_p \gamma_1^2+2 \gamma_2^2+2N_p \gamma_1\gamma_2) \sum_{m=1}^d c_{m}^2,
				\end{aligned}
			\end{equation*}
			where $c_{m}=\displaystyle\max_{m}\left(r-\frac{\sigma_m^2}{2}, \frac{\sigma_m^2 \rho_{m}}{2}\right).$
		\end{proof}

		\begin{lemma}\label{gate_d}
			The approximated time evolution operator ${\bf V}_{BS}(\tau )$ in \eqref{dvbs} can be implemented
			using ${\cal Q}_{single}=O(d N_pn_x)$ single-qubit gates and at most ${\cal Q}_{V_{BS}} = O(d N_pn_x^2)$ {\rm CNOT} gates for
			$n_x \geq 3$.
		\end{lemma}
		
		\begin{proof}
			Based on the same discussion in the proof of
			Lemma \ref{gate}, the single-qubit gates for ${\bf V}_1$ and ${\bf V}_2$ can be directly implemented  for $d$ dimension. 
			\begin{equation*}
				\begin{aligned}
					{\cal Q}_{single-d}=&{\cal Q}_{single-V_1}+{\cal Q}_{single-V_2}\\
					=&d 2^{n_p-1}  (2+2n_x+2n_x)
					+d(4n_x)
					=O(d N_p n_x).
				\end{aligned}
			\end{equation*}
			Similarly, for CNOT gate, we have 
			\begin{equation*}
				\begin{aligned}
					{\cal Q}_{\bf{BS}}=d{\cal Q}_{V_2} +d 2^{n_p-1} {\cal Q}_{V_1}+d(2^{n_p}-1){\cal Q_{c-V_1}}
					=O(d N_p n_x^2).
				\end{aligned}
			\end{equation*}
		\end{proof}

		\begin{lemma}\label{lem2}
			Let \( {\bf H}_{BS} \) be the Hamiltonian defined in Equation \eqref{dhbs}. The time evolution operator \( {\bf U}_{BS}(T) = \exp(i{\bf H}_{BS}T) \) for a time duration \( T \) can be implemented on a \( (dn_x + n_p) \)-qubit system using quantum circuits with
			$
			O\left(\frac{1}{\varepsilon}d N_p T^2 n_x^2 (N_p \gamma_1^2 + 2 \gamma_2^2 + 2N_p \gamma_1\gamma_2)\sum_{m=1}^d c_{m}^2\right)
			$
			single-qubit gates and
			$
			O\left(\frac{1}{\varepsilon} d N_p n_x^3 T^2 (N_p \gamma_1^2 + 2 \gamma_2^2 + 2N_p \gamma_1\gamma_2)\sum_{m=1}^d c_{m}^2\right)
			$
			CNOT gates, within an additive error of \( \varepsilon \).
		\end{lemma}
		
		\begin{proof}
			To ensure that the error in the simulation over the total time \( T \) is bounded by \( \varepsilon \), it suffices to divide the total time \( T \) into \( N \) intervals, where \( N = T/\tau \), such that
			\begin{equation}\label{errd}
				\|{\bf U}_{BS}^N(\tau) - {\bf V}_{BS}^N(\tau)\| \leq N \|{\bf U}_{BS}(\tau) - {\bf V}_{BS}(\tau)\| \leq \varepsilon.
			\end{equation}
			From Equation \eqref{err_d}, we have
			\begin{equation*}
				\frac{1}{4N} T^2 n_x (N_p \gamma_1^2 + 2 \gamma_2^2 + 2N_p \gamma_1\gamma_2) \sum_{m=1}^d c_{m}^2 \leq \varepsilon,
			\end{equation*}
			which can be rearranged to yield
			\begin{equation*}
				N \geq \frac{1}{4\varepsilon} T^2 n_x (N_p \gamma_1^2 + 2 \gamma_2^2 + 2N_p \gamma_1\gamma_2) \sum_{m=1}^d c_{m}^2.
			\end{equation*}
			By Lemma \ref{gate_d}, the unitary \( {\bf V}_{BS} \) can be implemented using \( O(d N_p n_x) \) single-qubit gates and at most \( O(d N_p n_x^2) \) CNOT gates.
			
			Thus, the implementation of \( {\bf V}_{BS}^N \) requires
			$
			O\left(\frac{1}{\varepsilon}d N_p T^2 n_x^2 (N_p \gamma_1^2 + 2 \gamma_2^2 + 2N_p \gamma_1\gamma_2)\sum_{m=1}^d c_{m}^2\right)
			$
			single-qubit gates and
			$
			O\left(\frac{1}{\varepsilon} d N_p n_x^3 T^2 (N_p \gamma_1^2 + 2 \gamma_2^2 + 2N_p \gamma_1\gamma_2)\sum_{m=1}^d c_{m}^2\right)
			$
			CNOT gates.
		\end{proof}

	\begin{theorem}
		For the d-dimensional Black-Scholes equation  \eqref{bsd}, the state $|{\bf U}(t)\rangle$, where ${\bf U(t)}$ is the solution of
		Equation \eqref{eq:d-dim} with a mesh size $h$, can be prepared up to precision $\varepsilon'$ using the Schr\"{o}dingerisation
		method depicted in Figure \ref{VBSd} and Figure \ref{schro}. This preparation can be achieved using at most
		$\tilde{O}\left(\frac{d T^2 \|{\bf U}(0)\|^3 \sum_{m=1}^d c_{m}^2}{h^4\varepsilon^3 \|{\bf U}(T)\|^3}\right)$
		single-qubit gates and {\rm CNOT} gates.
	\end{theorem}
	
	\begin{proof}
		As discussed in Theorem~\ref{th1}, the computational complexity of the Schrödingerisation approach for the 
		d-dimensional Black-Scholes equation with central difference discretization consists of three key components.
		
		The complexity analysis of simulating the unitary ${\bf V}_{BS}^N(\tau)$
		as concluded in Lemma \ref{lem2} can be
		applied with $n_x=O(\log(L_x/h))$ and $\gamma_1=\frac{1}{h^2 L_p}, \gamma_2=\frac{1}{2h}$.
		This results in a computational
		cost of  $O\left(\frac{dN_p^2 T^2 \log^3(L_x/h)\sum_{m=1}^d c_{m}^2}{h^4 L_p^2 \varepsilon'}\right)$,
		with $\varepsilon'=\|{\bf U}(T)\| \varepsilon /\|{\bf U}(0)\|$.
		Following Theorem~\ref{th1}, the required number of measurements and the error analysis of the output state preserve the same scaling.
		Then	combining these components, we obtain the total complexity,
		$
		\tilde{O}\left(\frac{d T^2 |{\bf U}(0)|^3 \sum_{m=1}^d c_{m}^2}{h^4\varepsilon^3 |{\bf U}(T)|^3}\right).
		$

	\end{proof}
		
		\begin{remark}
			Introducing the inter-dimensional coupling (correlation coefficients \textit{$\rho_{mn}\neq 0,$ $m\neq n$}) requires additional quantum gates to implement the non-diagonal operations. If cross terms are considered, it is necessary to simulate the non-diagonal operator $\partial x_i \partial x_j$  which may introduce $O(d^2)$ terms (similar to the high-dimensional coupling in the heat equation \cite{hu2024}). 
			
		\end{remark}

		\begin{remark}
			The Schr\"{o}dingerisation method provides a promising pathway for quantum simulation of the Black-Scholes equation, particularly in high-dimensional settings such as multi-asset option pricing.  Classical finite-difference methods suffer from the curse of dimensionality, requiring $( O(N_x^d) )$ for $d$ underlying assets. In contrast, the quantum approach scales as $\tilde{O}(d/\varepsilon^3)$, offering exponential memory savings and polynomial gate complexity in $d$. This advantage becomes pronounced for $d\geq 3$, where classical methods become computationally prohibitive.  
		\end{remark}
		\section{Numerical results}
		
		In this section, we use the proposed quantum circuit  by  the Qiskit package \cite{qiskit_book} and  the classical difference method to verify the algorithm in the framework of the  Black-Scholes model to ensure the feasibility of Schr\"{o}dingerisation algorithm. Throughout the implementation
		phase, the {\it Statevector} simulator is utilized
		to retrieve the complete solution, while the
		{\it Estimator} function is employed for the measurement of observable.
		\begin{example}
			The classic European call option price has an explicit expression at $t$ time:
			\begin{equation*}
				u(x,t)=e^x N(d_1)- K e^{-r t} N(d_2),
			\end{equation*}
			where
			\begin{equation*}
				d_1=\frac{x-\ln K+(r+\frac{1}{2}\sigma^2)t}{\sigma \sqrt{t}},\quad d_2=d_1-\sigma \sqrt{t},\quad 
				N(y)=\frac{1}{2\pi} \int_{-\infty}^{y} e^{-\frac{1}{2}s^2} ds.
			\end{equation*}
		\end{example}
		To formulate the problem,	parameters of  equation are given: $T=1$, rate $r=0.02$,
		volatility $\sigma=0.3$, exercise price $K=30$, 
		minimum asset price $s_{\min}=10^{-4}$, maximum asset price $s_{\max}=10K$, and the logarithm of  			$x\in L_x= [x_{\min}, x_{\max}]=[\ln(10^{-4}),\ln(10K)]$,
		auxiliary variable $p\in L_p=[p_{\min}, p_{\max}]=[-4\pi,4\pi]$.
		
		In this test, we choose  two forms of initial values to simulate the problem.
		The first one is shown in  \eqref{initial} which is continuous
		but not in $C^1$. 
		In order to increase the order of convergence of variables $p$,  we consider a   smoother initial solution discussed in 
		\cite{shima2024},
		\begin{align*}
			g(p)=\left\{
			\begin{array}{l}
				(-3+3e^{-1} )p^3+	(-5+4e^{-1} )p^2-p+1, \quad p\in(-1,0),\\
				e^{-p}, \quad p\in (-\infty,-1] \cup [0,\infty),
			\end{array}
			\right.\label{smooth}
		\end{align*}
		and initial value of \eqref{warp} is	${\bf v}(0):= \bar{\bf u}(0) \otimes {g(p)}$.
		The  numerical results of two different initial solutions are denoted as
		$u_h$ and $u_h^*$, respectively.
		
		Denote the error as $\epsilon = \vert u_{} - u_{numerical}\vert$, where $u_{numerical}$ denoted as $u_h$ and $u_h^*$. In finance and real life applications, people focus on the function value evolution when $S$ is in the neighborhood of the point $S = E$. Thus in the following representation of numerical results, we throw away the boring points and focus on the region that actually matters. In fact, we show the behavior of the function in the interval $S\in [0, 2E]$. We also show the error of the classical method and the Trotter splitting method when the same number of points are used to discretize the space dimension $x$. 
		Our numerical experiments demonstrate that employing smoother initial conditions along the $p$-dimension yields solutions with significantly improved agreement to the analytical solution.
		\begin{figure}[!htp]
			\centering
			\begin{minipage}{0.48\textwidth}
				\centering
				\includegraphics[height=5.7cm]{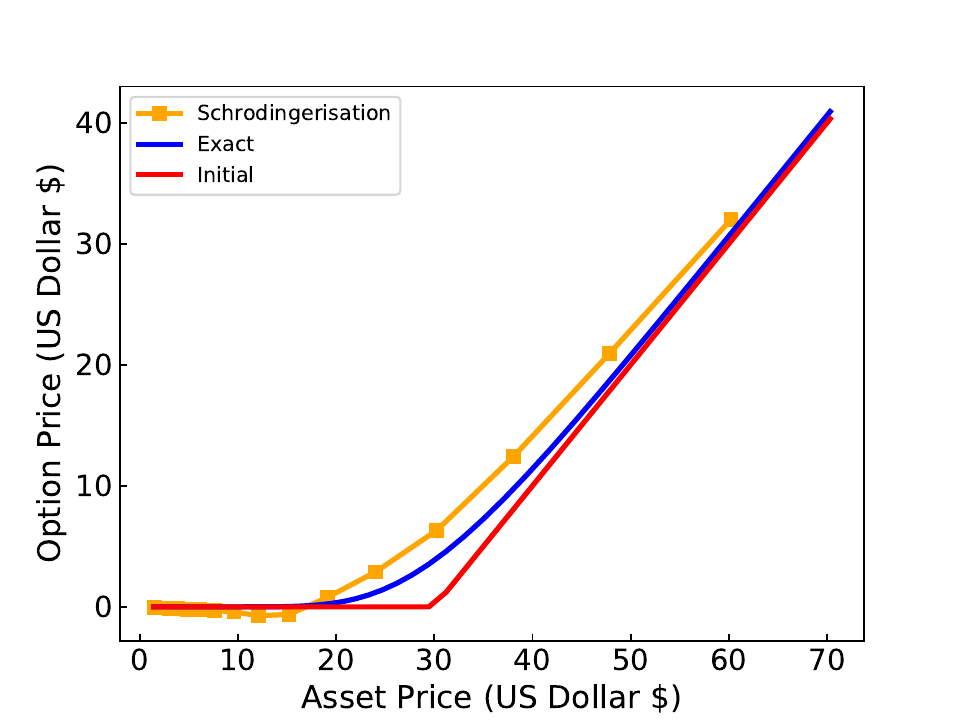}
				
			\end{minipage}\hfill
			\begin{minipage}{0.48\textwidth}
				\centering
				\includegraphics[height=5.7cm]{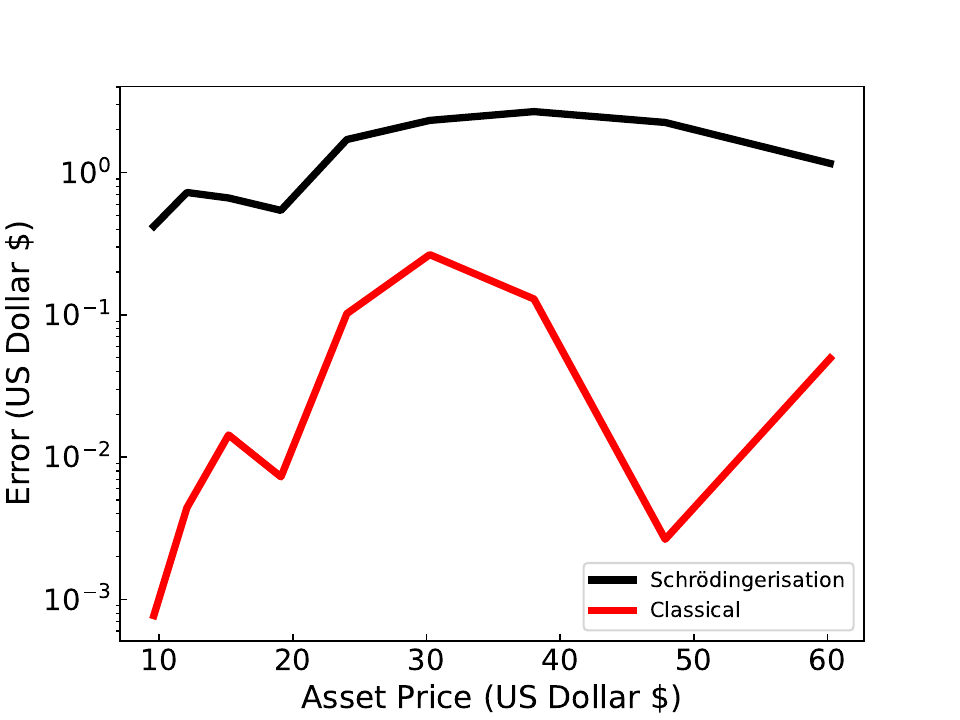}
			\end{minipage}
			\caption{\centering Figures of the numerical solution $u_h$ and error when $n_{x}=6,\ n_{p} = 4$.
				CPU runtime: 00:20:17. Maximum RAM occupancy: 5.47 GB} 
			\label{fig7}
		\end{figure}
		
		\begin{figure}[!htp]
			\centering
			\begin{minipage}{0.48\textwidth}
				\centering
				\includegraphics[height=5.7cm]{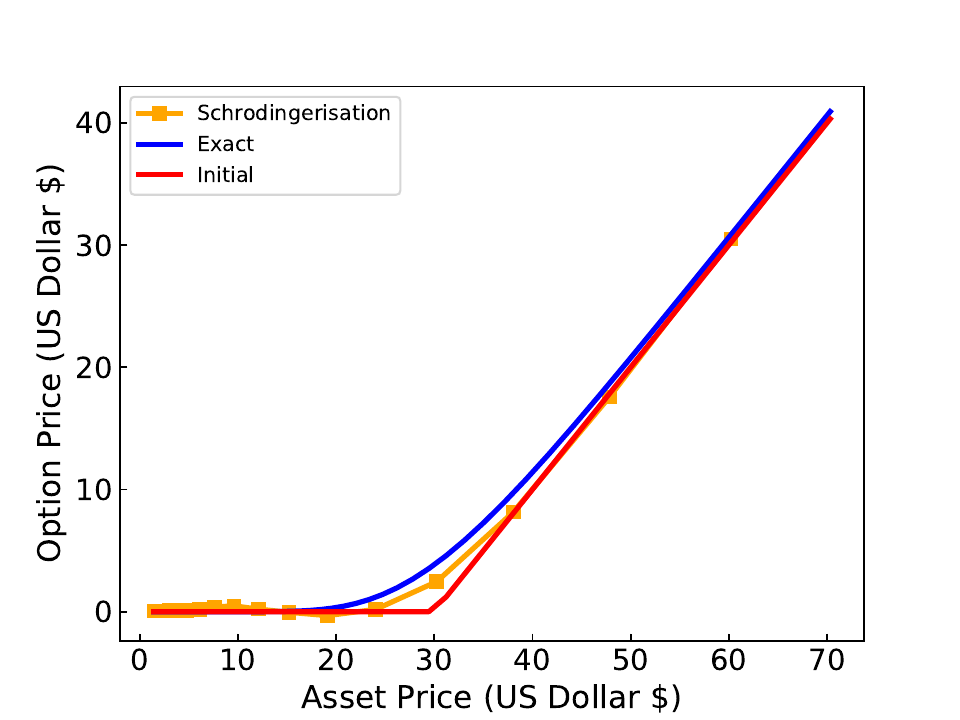}
				
			\end{minipage}\hfill
			\begin{minipage}{0.48\textwidth}
				\centering
				\includegraphics[height=5.7cm]{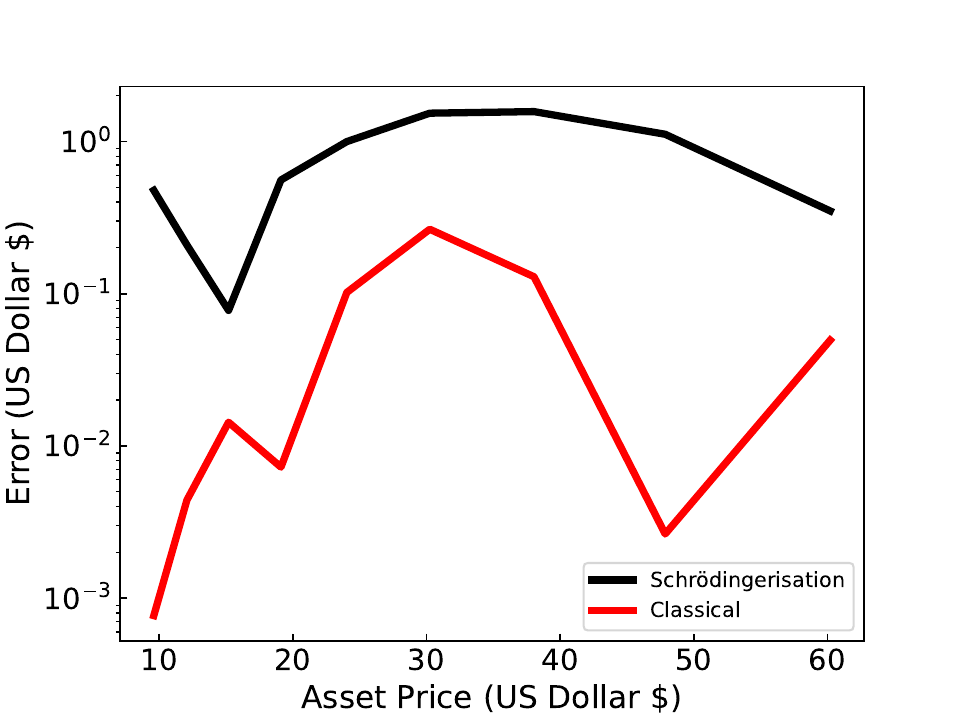}
			\end{minipage}
			
			\caption{\centering Figures of the numerical solution  $u_h$  and error when $n_{x}=6,\ n_{p} = 5$. 
				CPU runtime: 01:11:15. Maximum RAM occupancy: 6.09 GB} 
			\label{fig8}
		\end{figure}
		
		\begin{figure}[!htp]
			\centering
			\begin{minipage}{0.48\textwidth}
				\centering
				\includegraphics[height=5.7cm]{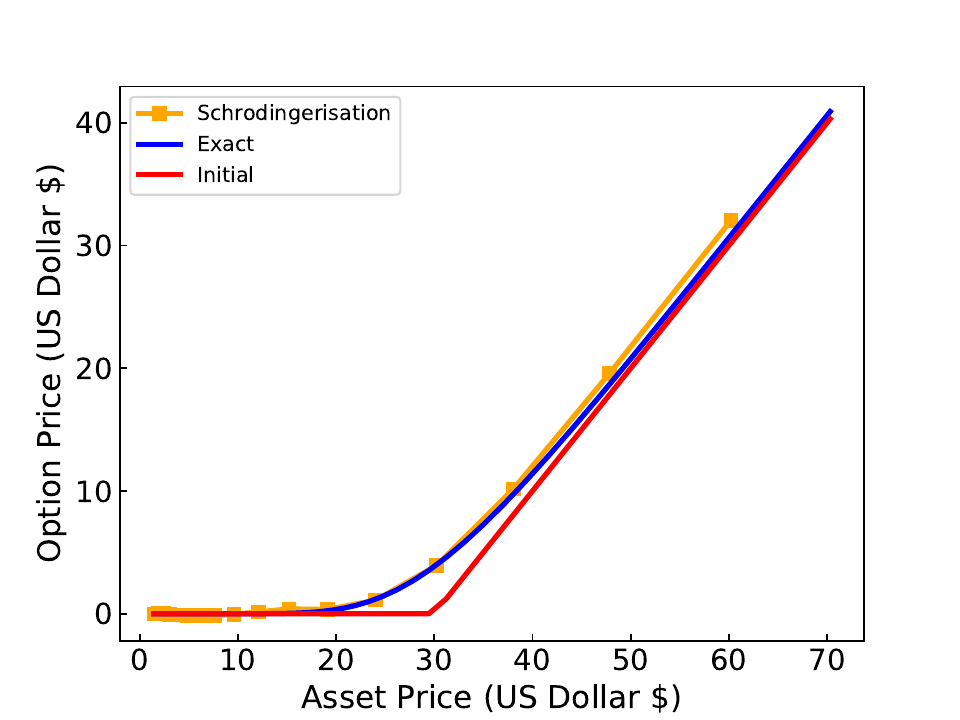}
				
			\end{minipage}\hfill
			\begin{minipage}{0.48\textwidth}
				\centering
				\includegraphics[height=5.7cm]{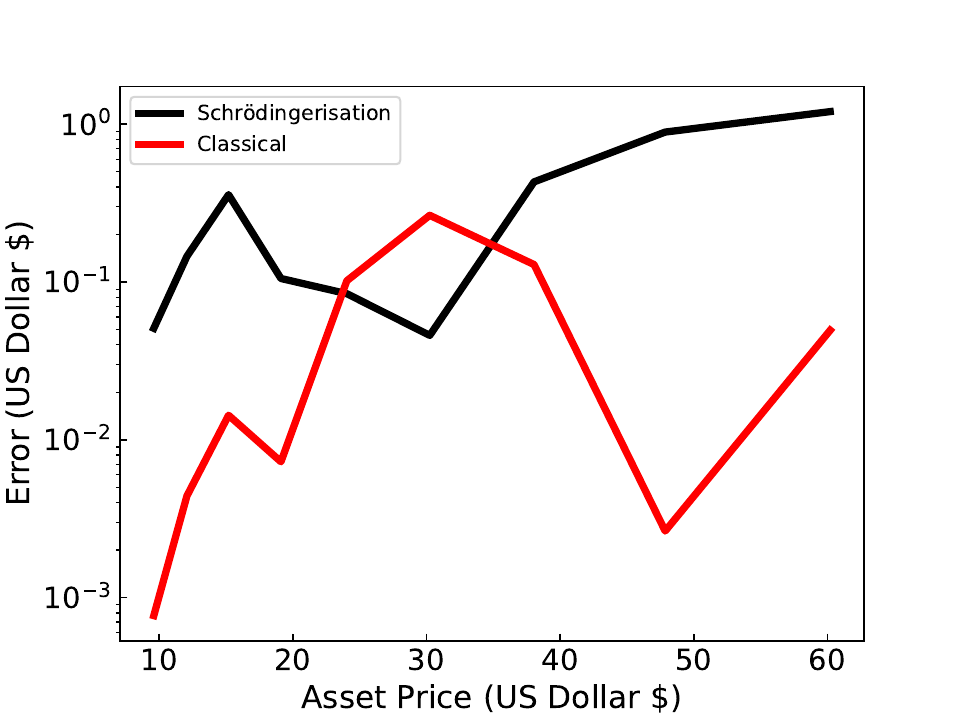}
			\end{minipage}
			
			\caption{\centering Figures of the numerical solution  $u_h$  and error when $n_{x}=6,\ n_{p} = 6$.
				CPU runtime: 03:37:11. Maximum RAM occupancy: 7.12 GB} 
			\label{fig9}
		\end{figure}
		\newpage
		
		\begin{figure}
			\centering
			\begin{minipage}{0.48\textwidth}
				\centering
				\includegraphics[height=5.7cm]{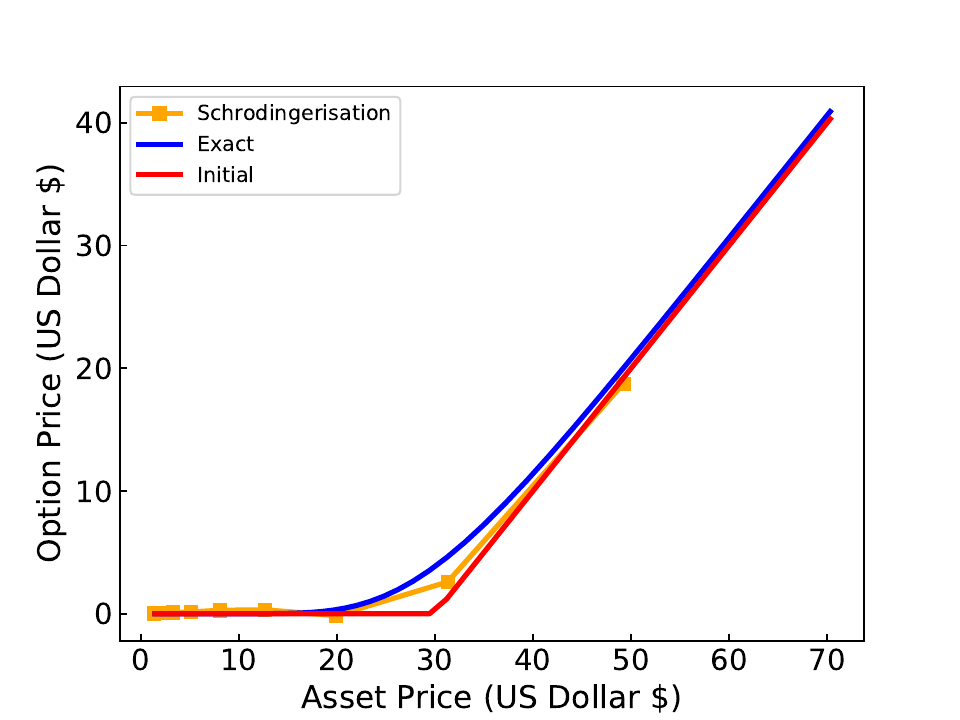}
				
			\end{minipage}\hfill
			\begin{minipage}{0.48\textwidth}
				\centering
				\includegraphics[height=5.7cm]{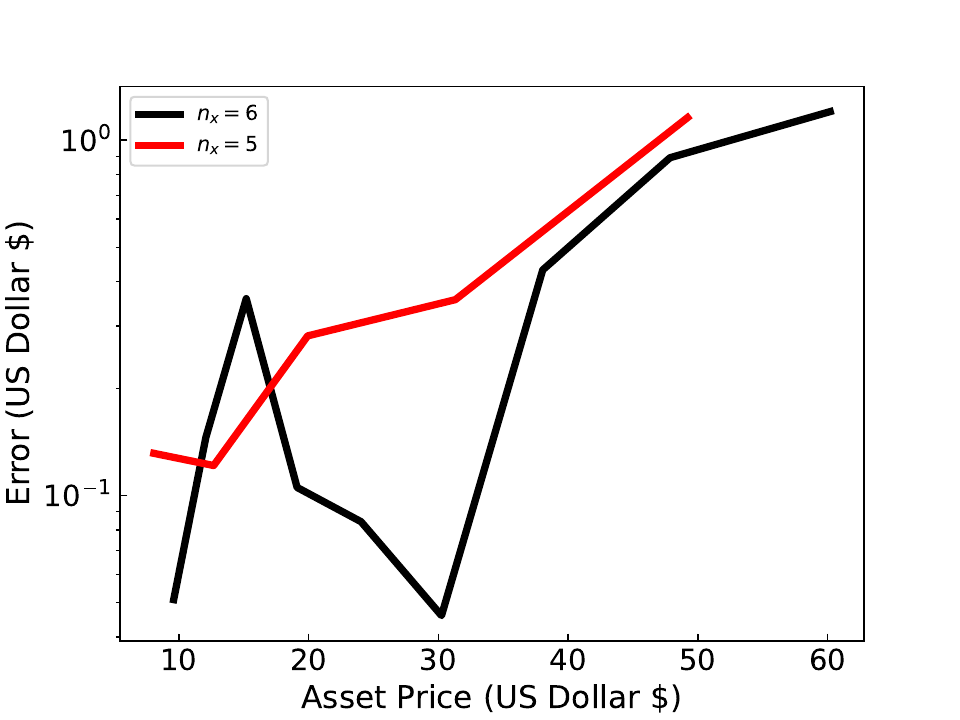}
			\end{minipage}
			
			\caption{\centering Figures of the numerical solution  $u_h$  and error when $n_{x}=5,6$ and $ n_{p} = 6$.
				CPU runtime: 00:18:08. Maximum RAM occupancy: 5.36 GB} 
			\label{fig10}
		\end{figure}
		
		\begin{figure}
			\centering
			\begin{minipage}{0.48\textwidth}
				\centering
				\includegraphics[height=5.7cm]{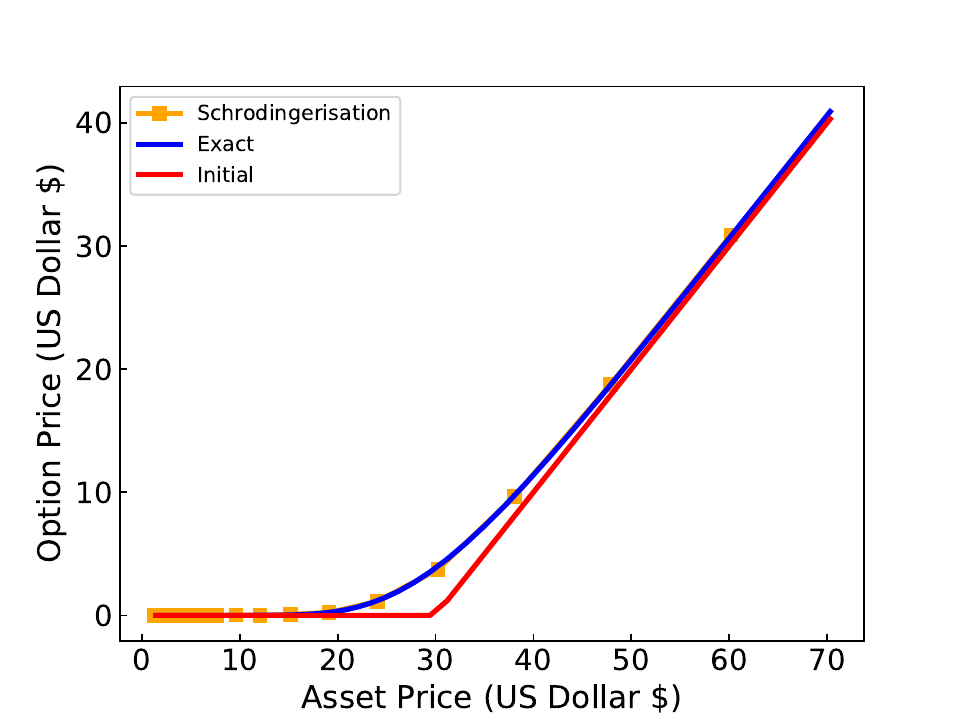}
				
			\end{minipage}\hfill
			\begin{minipage}{0.48\textwidth}
				\centering
				\includegraphics[height=5.7cm]{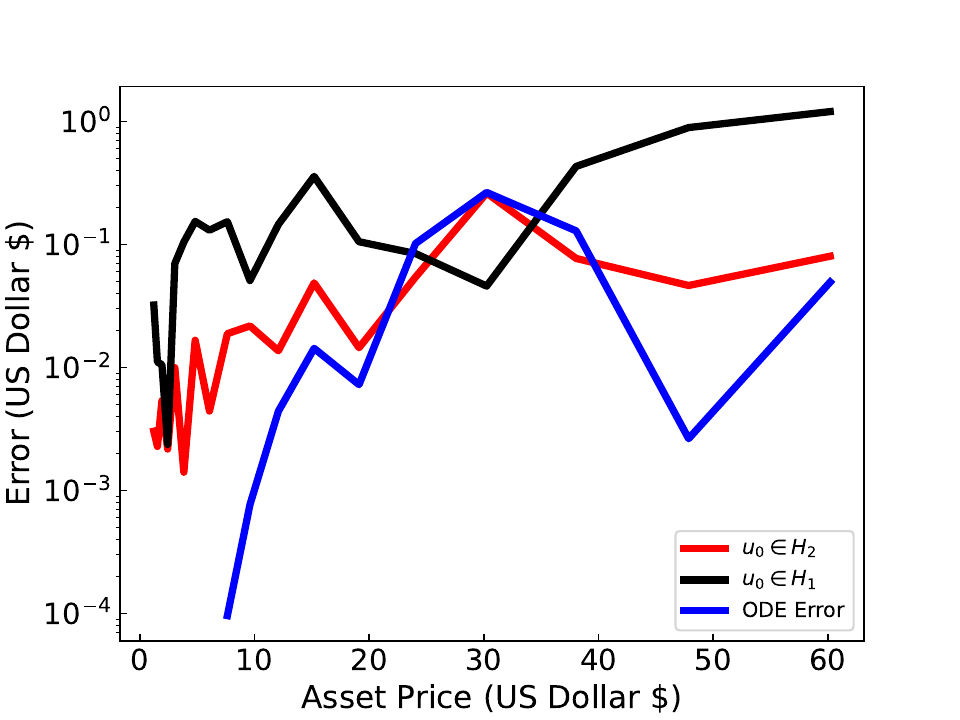}
			\end{minipage}
			
			\caption{\centering Figures of the numerical solution  $u_h^*$  and error  when $n_{x}=6,\ n_{p} = 6$.
			} 
			\label{fig11}
		\end{figure}

		To be clarified, the classical error in the figures comes from the directly solving of ODE \eqref{finite} with matrix exponential. The point number used to discretize the space dimension in the original PDE is the sames as the number set in Schrödingerisation procedure.  We use such classical error as a reference line and compare the right sub-figure in Figure \ref{fig7}, Figure \ref{fig8} and Figure \ref{fig9}, one can tell that as the number of ancilla qubits $n_{p}$ increases, the error produced by Schrödingerisation procedure steadily decrease. When $n_{p} = n_{x} = 6$, according to the analysis in \cite{jin2023}, the relative error $\epsilon = \frac{\Vert w_{h} - w\Vert}{\Vert w\Vert}$ generated by the discretization on p dimension satisfies : $\epsilon \sim\Delta p = \pi/8$. $w$ is the exact solution of ODE \eqref{finite}. Two pictures in Figure \ref{fig9} verified the theory.  Furthermore, comparing the left sub-figure in  Figure \ref{fig7}, Figure \ref{fig8} and Figure \ref{fig9}, one can also conclude that with the increase of $n_{p}$, numerical solution corresponding to the Schrödingerisation method converges to the exact solution. The convergence underscores the capability of the Schrödingerisation method to effectively approximate the solution. In real life applications, people are also concerned with where the inflection point of the function $V$ with respect to $S$ occurs at time $t$. For example, the asset price $S = E$, $E = 30$ in our example, is the inflection point when $t = 0$. Such point moves backward as $t$ increase. One can tell from the figures that  the Schrödingerisation method catches the inflection precisely. However due to the limitation of computing resources, the simulation with parameters $n_{x} = 6,\ n_{p} = 6$  is the maximum computational load that the computer can handle. Therefore we do the simulation when $n_{x} = 5,\ n_{p} = 6$ to compare with the best cases. Figure \ref{fig10}  shows that when keeping $n_{p}$ unchanged, more points on $x$ dimension leads to better accuracy. Figure \ref{fig11} shows that when we choose smoother initial condition ${\bf v}(0):= \bar{\bf u}(0) \otimes {g(p)}$ , the error will further decrease.
		
		In order to further illustrate the effectiveness of the Schr\"{o}dingerisation algorithm and the effect of initial value, the convergence rates of  two initial values with $\Delta t=10^{-3}, T=0.1$ are shown in Table \ref{rate:1d}. The results are obtained by using the central difference method for the Schr\"{o}dingerisation method.
		According to the Table, second-order convergence rates of smooth initial value  ${\bf v}(0):= \bar{\bf u}(0) \otimes {g(p)}$ can be obtained.

		
		
		\begin{table}
			\begin{center}
				\begin{tabular}{ c c c   c c c c}
					\hline
					$(\Delta p, \Delta x)$ & $(\frac{L_p}{2^7}, \frac{L_x}{2^6})$& order & $(\frac{L_p}{2^8}, \frac{L_x}{2^7})$& order & $(\frac{L_p}{2^9}, \frac{L_x}{2^8})$& order\\ 
					\hline
					$\|u_h-u \|_{L^2}$ &6.698e-01 & - & 3.271e-01 &1.03 &1.776e-01 &   0.89\\  
					\hline
					$\|u_h^*-u \|_{L^2}$&2.925e-01&-  & 7.540e-02 &1.96    &2.060e-02& 1.88\\ 
					\hline
				\end{tabular}
			\end{center}
			\caption{  The convergence rates of two initial solution $\|u_h-u \|_{L^2}$ and 	$\|u_h^*-u \|_{L^2}$, using the classical finite difference method and the Schr\"{o}dingerisation method.}
		\end{table}\label{rate:1d}
		
		\begin{example}
			Consider the two-dimensional Black-Scholes pricing model for cash or valueless options
			$$
			\left\{
			\begin{array}{l}
				\frac{\partial u}{\partial \tau }=\frac{1}{2}
				\sigma_1^2 \frac{\partial^2 u}{\partial{x}^2}
				+\frac{1}{2}
				\sigma_1^2 \frac{\partial^2 u}{\partial{y}^2}+
				\sigma_1 \sigma_2 \rho_{12} \frac{\partial^2 u}{\partial{x}\partial {y}}+ \left(r-\frac{\sigma_1^2}{2}\right)\frac{\partial u}{\partial x}+ \left(r-\frac{\sigma_2^2}{2}\right)\frac{\partial u}{\partial y} -ru,\\
				u(x,y,0)=u_0(x,y),
			\end{array}
			\right.
			$$
			where $u_0(x,y)=c$ with $x>\ln(K_1), y>\ln(K_2)$, otherwise $u_0(x,y)=0$.
		\end{example}
		Assume the amount paid by the right at maturity $c=1$,
		exercise price $K_1=K_2=50$, risk-free interest rate $r=0.03$, asset volatility $\sigma_1=\sigma_2=0.3$,
		asset correlation coefficient $\rho=0.6$,
		the logarithm of  asset price
		$x\in [x_{\min}, x_{\max}]=[\ln(10^{-6}),\ln(4K_1)]$,
		$y\in [y_{\min}, y_{\max}]=[\ln(10^{-6}),\ln(4K_2)]$,
		and	auxiliary variable $p\in L_p=[-\pi L,\pi L]$, $L=15$.
		Cash or valueless options satisfy the following boundary conditions
		\begin{equation*}
			\begin{aligned}
				u(x_{\min},y,\tau)&=u(x,y_{\min},\tau)=0,\quad \forall x,y,\tau,\\
				\frac{\partial u}{\partial x}(x_{\max},y,\tau)&=	\frac{\partial u}{\partial y}(x_{\max},y,\tau)=0,\quad
				\forall x,y,\tau. 
			\end{aligned}
		\end{equation*}
		and the exact solution is 
		\begin{eqnarray*}
			u(x,y,\tau)= c e^{-r\tau} B(d_x,d_y,\rho),
		\end{eqnarray*}
		where
		\begin{equation*}
			\begin{aligned}
				d_x&=\frac{x-\ln K_1+(r-\frac{1}{2}\sigma_1^2)\tau}{\sigma_1 \sqrt{\tau}},\\
				d_y&=\frac{x-\ln K_2+(r-\frac{1}{2}\sigma_2^2)\tau}{\sigma_2 \sqrt{\tau}},\\
				B(d_x,d_y,\rho)&=\frac{1}{2\pi \sqrt{1-\rho^2}} \int_{-\infty}^{d_x}\int_{-\infty}^{d_y} e^{-\frac{s_1^2-2\rho s_1 s_2+s_2^2}{2(1-\rho^2)}} ds_2 d s_1.
			\end{aligned}
		\end{equation*}
		The numerical results and  the images in the upper right corner are shown in the Figure \ref{fig:d} and \ref{fig:2d} using Matlab based on two initial values similar with  the one-dimensional  case.
		The results are obtained by using the central difference method for Schr\"{o}dingerisation with
		the parameters  	$T=1$, $\Delta t=10^{-3}$, $n_x=8$ and $n_p=7$.
		The numerical results show that the Schr\"{o}dingerisation method can simulate the price curve of dual-asset options well. 
		In this test, the Schr\"{o}dingerisation method can be directly applied to the Equation  \eqref{eq:d-dim} without dilation.
		Any value of $p$ that is slightly larger than $p^{\diamond}=\max(\lambda({\bf A}_1)T,0)$ can well retrieve the original solution.
		
		\begin{figure}[htbp]
			\centering
			\subfigure[$u_h$]{
				\begin{minipage}{5cm}
					\includegraphics[width=5cm]{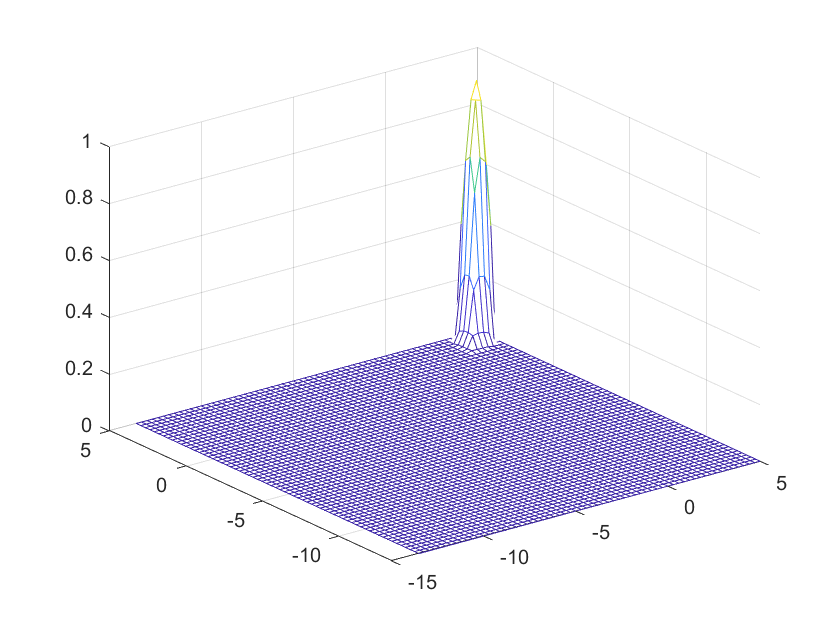}
				\end{minipage}
			}
			\subfigure[$u_h^*$]{
				\begin{minipage}{5cm}
					\includegraphics[width=5cm]{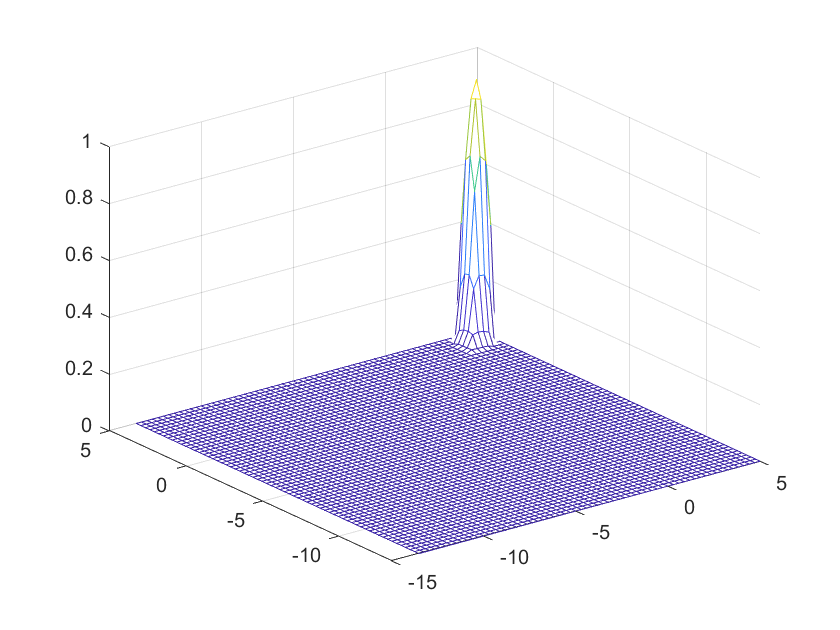}
				\end{minipage}
			}
			\caption{ The solutions of the 2d Black Sholes equation with two 
				initial values .}
			\label{fig:d}
		\end{figure}
		
		\begin{figure}[htbp]
			\centering
			\subfigure[$u_h$]{
				\begin{minipage}{5cm}
					\includegraphics[width=5cm]{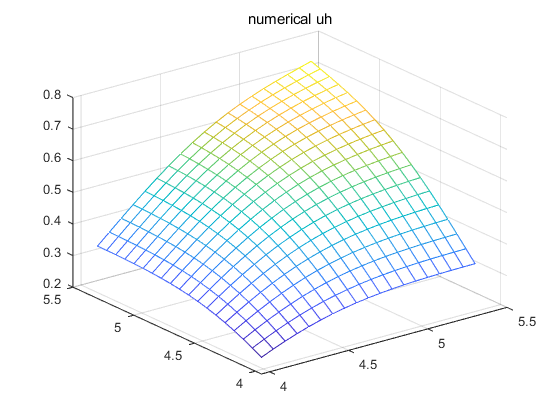}
				\end{minipage}
			}
			\subfigure[$u_h^*$]{
				\begin{minipage}{5cm}
					\includegraphics[width=5cm]{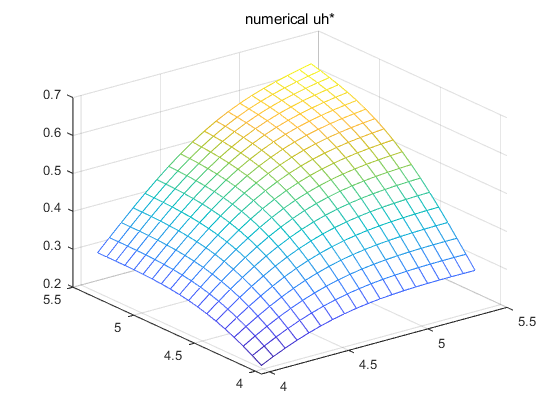}
				\end{minipage}
			}
			\caption{ The solutions of 2d Black Sholes equation (upper right).}
			\label{fig:2d}
		\end{figure}											
		
		\section{Conclusion}
		
		In this paper, we develop a quantum computing framework for solving the Black-Scholes equation by combining the Trotter splitting method with the Schr\"odingerisation approach. We present a practical implementation of quantum circuits based on operator splitting techniques, including detailed circuit designs for each computational component. Our analysis provides comprehensive quantum circuit descriptions and complexity estimates, demonstrating quantum advantages in high-dimensional settings compared to classical methods.
		
		Numerical experiments performed on the Qiskit platform validate the convergence of our Schr\"odingerisation method for both one-dimensional and two-dimensional Black-Scholes equations, confirming its effectiveness for financial derivative pricing applications.
		
		\section*{Acknowledgement}
		SJ and LZ  were  supported by  NSFC grant No. 12341104, the Shanghai Jiao Tong University 2030 Initiative, the  Science and Technology Innovation Key R\&D Program of Chongqing grant No. CSTB2024TIAD-STX0035, and the Fundamental Research Funds for the Central Universities.  SJ was also supported by NSFC grant Nos. 12350710181 and 12426637. XY was supported by NSFC grant No. 12426626. 										
		\bibliographystyle{plain}
		\normalem
		\bibliography{refs}
		
	\end{document}